\documentclass[11pt]{article}

\usepackage[paper=letterpaper, margin=1in]{geometry}

\usepackage{ifthen}

\usepackage{amsmath}
\usepackage{amssymb}

\usepackage{amsthm}

\usepackage{authblk}

\usepackage{graphicx}

\usepackage[%
  ocgcolorlinks,%
  linkcolor=blue,%
  filecolor=blue,%
  citecolor=blue,%
  urlcolor=blue]{hyperref}

\usepackage{enumitem}

\usepackage{multirow}
\usepackage[singlelinecheck=true, labelfont=bf]{caption}

\usepackage[ruled]{algorithm2e} 

\SetAlFnt{\small}
\SetAlCapFnt{\small}
\SetAlCapNameFnt{\small}
\SetAlCapHSkip{0pt}
\IncMargin{-\parindent}

\usepackage{cleveref}
\usepackage{xfrac}

\newtheorem{theorem}{Theorem}[section]
\newtheorem{lemma}[theorem]{Lemma}

\newtheorem{corollary}[theorem]{Corollary}

\theoremstyle{definition}
\newtheorem{definition}[theorem]{Definition}

\newtheorem{remark}[theorem]{Remark}

\newtheorem*{definition*}{Definition}
\theoremstyle{plain}

\DeclareMathOperator*{\argmax}{arg\,max}

\newcommand{\makeset}[1]{\lbrace #1 \rbrace}
\newcommand{\setminussingle}[2]{#1 \setminus \makeset{#2}}
\newcommand{\setunionsingle}[2]{#1 \cup \makeset{#2}}
\newcommand{\marginal}[2]{f(#2: #1)}
\newcommand{\XOS}{\texttt{XOS}}
\newcommand{\EQref}[1]{Eq.~\eqref{#1}}
\newcommand{\DSIC}{\texttt{DSIC}}
\newcommand{\BB}{\texttt{BB}}

\newcommand{\ev}{\text{eligible values}}
\newcommand{\ANY}{\circ}
\newcommand{\RULE}[1]{\mathbf{#1}}

\newcommand{\SC}{\texttt{SATOSH}}
\newcommand{\SEL}{\RULE{S}}
\newcommand{\AHT}{\texttt{AHT}}

\makeatletter
\newcommand*{\rom}[1]{\expandafter\@slowromancap\romannumeral #1@}
\makeatother

\usepackage[dvipsnames]{xcolor}

\usepackage{mathtools}
\newcommand{\defeq}{\coloneqq}
\hypersetup{
  hypertexnames=false,
  pdftitle={BROCADE: Broker-based Contract Design for the Multi-Agency Model},
  pdfauthor={Yangguang Shi, Ron Lavi, Jie Zhang}
}

\begin{document}

\title{BROCADE: Broker-based Contract Design for the \\Multi-Agency Model}

\author[1]{Yangguang Shi}
\author[2]{Ron Lavi}
\author[3]{Jie Zhang}
\affil[1]{School of Computer Science and Technology, Shandong University, China\\
\texttt{shiyangguang@sdu.edu.cn}}
\affil[2]{Department of Economics, University of Bath, United Kingdom\\
\texttt{arl65@bath.ac.uk}}
\affil[3]{Department of Computer Science, University of Bath, United Kingdom\\
\texttt{jz2558@bath.ac.uk}}

\date{}

\begin{titlepage}

\maketitle

\begin{abstract}
We study welfare maximization in a contract model with $m$ principals who privately value a public project's success and $n$ agents whose costly effort is unobservable. Contracts induce truthful dominant strategies for principals and equilibrium effort by agents, maximizing the expected total value of eligible principals.

We develop three tools: a broker that separates payments from principals and payments to agents, an extension of Myerson's Lemma characterizing dominant-strategy incentive compatibility (DSIC), and a black-box reduction to budget-constrained single-principal contract design. When the success probability is a strictly monotone XOS function of agents' actions, these tools yield three results.

First, a deterministic DSIC contract achieves a constant approximation against a benchmark that qualifies every principal and maximizes success probability using budget $\sum_j v^j$, but may violate budget balance (BB).

Second, against benchmarks with budgets at least the largest principal value, deterministic DSIC and BB contracts admit no finite approximation guarantee, even with two agents. Universally DSIC and universally BB randomized contracts have an $\Omega(n)$ approximation lower bound, even for a fixed number of principals.

Third, using one quarter of the optimal fixed-price revenue supported by at least two principals as the benchmark budget, we obtain a universally DSIC and universally BB randomized contract with an $O(\log m)$ approximation ratio independent of $n$. 
\end{abstract}

\renewcommand{\thepage}{}
\end{titlepage}

\clearpage

\pagenumbering{arabic}

\section{Introduction}

Contract theory studies how to incentivize agents to take costly actions when their actions cannot be directly observed. Such problems appear frequently in the Internet economy. For example, in online labor markets, freelancers need to exert effort to complete tasks from customers. Social-media influencers spend effort to promote different brands. In these applications, effort is costly and usually cannot be fully observed by the parties who benefit from it. The principal-agent model provides a natural game-theoretic framework for studying such interactions.

In many applications, a single principal needs to interact with a \emph{team} of agents. For example, there can be many freelancers in an online labor market and many influencers in a social network. Our work builds on the multi-agent binary-action model of~\cite{DuttingEFK23}. In this model, a single principal interacts with multiple agents, and each agent chooses whether to exert effort. The subset of agents who exert effort determines the probability that the principal's project succeeds. This model has received increasing attention in the recent algorithmic contract design literature~\cite{SODA2024,hann2024optimality,DuttingEFK25,ITCS24,FeldmanTzurPonitkaSchlesinger25,EC2025}.

However, in many Internet applications, the interaction can involve multiple principals as well. A freelancer may work for different customers, an influencer may promote multiple brands, and a platform may provide services to many sellers at the same time. In these examples, the actions of a group of agents can affect the values obtained by multiple principals. This motivates us to study contract design in a setting with both multiple principals and multiple agents.

In this paper, we study a multi-principal multi-agent binary-action model. There are $m$ principals and $n$ agents. Each agent can choose whether to exert effort, and exerting effort incurs a publicly known, strictly positive cost. The subset $S$ of agents who exert effort determines the success probability $f(S)$ of a common project. We consider the hidden-action setting, where the actions of agents cannot be observed by the principals. Each principal $j$ has a private value $v^j$ for the success of the project and reports a bid to the contract.

One feature of our model is the eligibility variable $\delta^j\in\{0,1\}$ for every principal $j$. If $\delta^j=1$, principal $j$ obtains her value $v^j$ when the project succeeds. If $\delta^j=0$, she does not obtain any value from the success. Therefore, although all principals face the same project outcome and the same success probability, the contract can decide which principals can benefit from the success. This setting can be viewed as a middle ground between private and public goods. For a private good, different principals can receive different outcomes. For a public good, all principals benefit from the same outcome. In our model, the success of the project is common, but the benefit from the success can be restricted to a subset of principals. We refer to such an outcome as a \emph{semi-public outcome}.

The interaction can be viewed as a two-round sequential game. In the first round, the principals simultaneously submit their bids. In the second round, the agents simultaneously decide whether to exert effort according to the contract determined by the bids. We would like the truthful bid to be a dominant strategy for every principal, while the agents choose an equilibrium action profile in the second round. Our objective is to maximize the expected total value of the eligible principals,
$\left(\sum_{j\in[m]}v^j\delta^j\right)f(S).$
When there is only one principal, this objective is closely related to the success-probability maximization objective studied in previous work on multi-agent contracts.

Designing a contract in this model has difficulties from both mechanism design and contract theory. On the principals' side, their values are private, so the contract needs to provide incentives for truthful bidding. On the agents' side, their actions are hidden, so the contract needs to provide incentives for effort. Moreover, these two parts are connected: the payments collected from the principals determine the resources that can be used to incentivize the agents, while the actions of the agents determine the success probability and therefore affect the utilities and incentives of the principals.

\paragraph{Our technical tools.}
We develop three main technical tools for handling these difficulties.

Our first tool is a \emph{broker-based contract} framework, which we call \emph{BROker-based ContrAct DEsign} (BROCADE). We introduce a non-strategic broker between the principals and the agents. The broker collects payments from the principals and uses these payments to incentivize the agents. Without a broker, the contract needs to specify a transfer between every pair of a principal and an agent, resulting in $O(mn)$ transfer variables. With the broker, we only need to consider transfers from the $m$ principals to the broker and transfers from the broker to the $n$ agents. This reduces the number of transfer variables to $O(m+n)$ and separates the principal side and the agent side of the contract. We also establish an equivalence between the broker-free and broker-based formulations under budget balance.

Our second tool is a characterization of the dominant-strategy incentive-com\-patible (\DSIC{}) condition for the principals. In our model, the allocation of a principal depends on both her eligibility and the success probability generated by the selected agents. Therefore, the effective allocation of principal $j$ is $\delta^j f(S).$
We extend the standard monotonicity argument in Myerson's Lemma to our setting. We show that an allocation rule is implementable if and only if this effective allocation is monotone in the bid of every principal, and we give the corresponding payment rule.

Our third tool is a black-box reduction for selecting the agents. Suppose that a budget $B$ is available to incentivize agents. Previous work provides constant-factor approximation algorithms for maximizing the success probability under such a budget constraint. However, directly using such an algorithm does not guarantee that the obtained success probability is non-de\-creasing when the budget increases. This monotonicity property is important for obtaining a \DSIC{} contract on the principals' side. We therefore develop a black-box transformation that converts a constant-approxi\-mation algorithm for a fixed budget into a \emph{budget-scalable} constant-approxi\-mation selection function. The construction uses a tagged partition of the possible budget range.

\paragraph{Our results.}
Using these tools, we obtain three main results when the success probability function $f$ is strictly monotone and \XOS{}.

Our first result is a deterministic and non-ex\-cludable contract, where all principals are eligible to benefit from the success of the project. The contract is \DSIC{} and achieves a constant-factor approximation with respect to the benchmark using $\mathtt{SumV}(v)=\sum_{j\in[m]}v^j$
as the budget for incentivizing the agents. This is a strong benchmark since it uses the total value of all principals to maximize the success probability. However, this contract may violate the budget-balance condition, because the total payment collected from the principals can be smaller than the budget used for incentivizing the agents.

Our second result establishes limits imposed by budget balance when the benchmark budget is at least the maximum principal value. No deterministic \DSIC{} and \BB{} contract has a finite approximation guarantee, even with two agents and an additive success probability function (Theorem~\ref{theorem: approximation impossbility with DSIC and BB}). For universally \DSIC{} and universally \BB{} random contracts, we prove an $\Omega(n)$ approximation lower bound, even for additive functions and any fixed number of principals (Theorem~\ref{theorem: impossiblity for universally truthful}). Thus a guarantee independent of the number of agents is impossible under these benchmarks.

Our third result achieves an approximation guarantee independent of $n$ by using a weaker benchmark budget. Let $\mathcal{F}^{2}(v)=\max_{k\geq 2} k\cdot \pi^{k}(v),$
where $\pi^{k}(v)$ is the $k$-th largest value among the principals. This quantity represents the optimal fixed-price revenue when at least two principals are considered. Based on the random-sampling idea of~\cite{FiatGHK02}, we develop a randomized contract called \emph{SAmpling TOll SHaring} (\SC{}). The contract randomly partitions the principals into two subsets. It uses the fixed-price revenue obtained from one subset to determine a toll for the other subset, and then uses this toll as the budget for selecting agents. We show that the resulting contract is universally \DSIC{} and universally budget-balanced. It achieves an $O(\log m)$ approximation, independent of $n$, with respect to the benchmark using $\frac{1}{4}\mathcal{F}^{2}(v)$ as the budget for incentivizing the agents.

Together, our results identify the role of the benchmark budget under budget balance. Without budget balance, a constant approximation is possible for the strong benchmark based on the total principal value. With budget balance, deterministic contracts admit no finite approximation guarantee under strong benchmarks, and universally truthful random contracts incur an $\Omega(n)$ loss. Using a quarter of the fixed-price revenue supported by at least two principals as the benchmark budget allows \SC{} to achieve an $O(\log m)$ approximation independent of the number of agents.

The broker framework also provides a way to connect algorithms for single-principal contract design with mechanism-design techniques for multiple principals. In particular, the broker separates the payments of the principals from the incentives of the agents, the monotonicity characterization handles truthful bidding, and the budget-scalable selection function allows us to use existing algorithms for budget-constrained contract design in the multi-principal setting.

\paragraph{Other related literature}

The extension of the basic single-agent, single-principal setting to multiple principals is known as the \emph{common agency} problem. With observable agent actions and a single agent, the problem of maximizing principals’ values was largely settled by \cite{bernheim1986menu}, and  \cite{prat2003games} generalize to a setting of multiple principals and multiple agents, still with observable actions. The paper \cite{LaviS19} instantiates the abstract concept of contractible contracts of \cite{peters2012definable} to study the price of anarchy in this complete-information game.
In contrast, with unobservable actions and incomplete information, the problem is still largely open. \cite{BernheimW86} first studied the common agency setting with unobservable actions and incomplete information, for the case of multiple principals and one agent. In the almost four decades that have passed since then, there has not been much advancement ever since. In the case of a single agent, \cite{alon2024incomplete} define and characterize a class of ``incomplete information VCG contracts (IIVCG)’’ and design a polynomial-time algorithm to determine whether a setting has an IIVCG contract which provides individual rationality for the principals and limited liability for the agent. They show that there exist instances in which such a contract does not exist.
With a single principal, several works have considered the design of contracts with incomplete information from a computational point of view~\cite{babaioff2006combinatorial,HoSV16,KleinbergK18,KR19,roughgarden2019approximately,DRT19,FOCS21}. Some of these works have multiple agents (e.g.,~\cite{XiaoWCTY20,AlonDPTT20}), or agent private types~\cite{GuruganeshSW20}.

We are not aware of any papers that consider multiple agents and multiple principals with private types.

\section{Preliminaries}\label{sec:preliminaries}

\paragraph{Multi-agency model.} This paper studies the contract design model of \emph{multi-agency} (multiple principals and multiple agents) \cite{Peters2003}. In this model, there are $m \geq 1$ principals delegating a project and interacting with $n \geq 1$ agents. Each agent $i \in [n]$ can choose to either exert effort or not. Under such a \emph{binary-action} setting, exerting effort incurs a strictly positive cost $c_i>0$ for every agent $i\in[n]$,
while not exerting effort has a cost of zero. We assume that each agent's cost is publicly known to all entities in the model.

The model studied in this paper focuses on the \emph{binary-outcome} case \cite{DuttingEFK23}, meaning the project delegated by the principals can either succeed or fail. A function $f: 2^{[n]} \mapsto [0, 1]$ maps the subset of agents choosing to exert effort to the probability that the project succeeds. We assume $f(\emptyset)=0$. The function $f$ is commonly known to all players. We adopt the classic \emph{hidden-action} setting, in which the principals can only observe the project's outcome, not the agents' actions. Each principal $j \in [m]$ has a value $v^{j} \geq 0$ for the success of the project, while her value is zero in case of a failure. The parameter $v^{j}$ is the \emph{private} information of each principal $j$, and is unknown to all the other entities. An instance of such a multi-agency model is denoted by $\mathcal{I} = \langle m, n, f, v, c \rangle$.

\paragraph{Contract and bids.} Since exerting effort is costly and the benefits of the project's success are obtained by principals, the agents have no interest in exerting effort. The \emph{contract} solves this moral hazard issue by defining transfers from principals to agents for incentivizing agents to exert effort.

Formally, a contract $\phi = \lbrace t, \delta\rbrace$ is a pair of mapping functions, $t: \mathbb{R}_{\geq 0}^{m} \mapsto \mathbb{R}_{\geq 0}^{m \times n}$ and $\delta: \mathbb{R}_{\geq 0}^{m} \mapsto \lbrace 0, 1\rbrace^{m}$, specified by the designer. On the side of the principals, the contract resembles a \emph{direct revelation mechanism} that requires each principal $j \in [m]$ to submit a \emph{bid} $b^{j} \in \mathbb{R}_{\geq 0}$.\footnote{Throughout, we use superscripts for principals and subscripts for agents.} A bid $b^{j}$ is said to be \emph{truthful} if $b^{j} = v^{j}$. The set of bids from all the principals is denoted by $b = \lbrace b^{j} \rbrace_{j \in [m]}$ and is called a bid profile. The contract $t$ takes the bid profile $b$ as input, and outputs a non-negative transfer $t_{j \to i}(b)$ from every principal $j$ to every agent $i$. The payment $t_{j \to i}(b)$ is being made only in case of success and in case of failure the payment is zero. Note that the payment in case of success can only depend on the bids and not on the set of agents who exerted effort since we have hidden actions.
The contract $\phi$ also outputs an eligibility profile $\delta$. For each principal $j \in [m]$, $\delta^{j} \in \lbrace 0, 1\rbrace$ indicates whether the valuation of $j$ can be realized through the success of the project.

If $\delta^j = 1$, principal $j$ obtains value $v^j$ upon success and makes payments according to the transfer rule $t$; otherwise, $j$ is excluded, receives no value, and makes no payment.

\paragraph{Utilities, equilibria, and incentive compatibility.} Fix the contract $\phi = \lbrace t, \delta\rbrace$. For each submitted bid profile $b$, a subgame $\mathcal{G} = \langle b, \phi\rangle$ is induced over the agents. Let $S \subseteq [n]$ be a subset of agents choosing to exert effort in $\mathcal{G} = \langle b, \phi\rangle$, then the expected utility of each agent $i \in [n]$ is given as
\begin{equation}
    u_{i}(b, \phi; S) = \begin{cases}
        \Big(\sum_{j \in [m]} t_{j \to i}(b)\Big) \cdot f(S) - c_{i} & \text{if}\; i \in S \\
        \Big(\sum_{j \in [m]} t_{j \to i}(b)\Big) \cdot f(S) & \text{if}\; i \notin S 
    \end{cases} \, . \label{eq: agent utility}
\end{equation}
The subset $S$ forms a (pure Nash) equilibrium in the subgame $\mathcal{G}$ if for every agent $i \in [n]$, it holds that
\begin{equation}
    u_{i}(b, \phi; S) \geq \begin{cases}
        \Big(\sum_{j \in [m]} t_{j \to i}(b)\Big) \cdot f(\setminussingle{S}{i}) & \text{if} \; i \in S\\
        \Big(\sum_{j \in [m]} t_{j \to i}(b)\Big) \cdot f(\setunionsingle{S}{i}) - c_{i} & \text{if} \,\ i \notin S
    \end{cases} \, . \label{eq: subgame equilibrium}
\end{equation}
When there are multiple equilibria in the subgame $\mathcal{G} = \langle b, \phi \rangle$, we assume that the agents always coordinate at the equilibrium $S(b, \phi)$ that favors the principals' welfare, which is discussed later, breaking welfare ties in favor of the highest success probability and any remaining ties by a fixed order.

Let $b$ be an arbitrary bid profile, and $S = S(b, \phi)$ be the equilibrium in the induced subgame $\mathcal{G} = \langle b, \phi\rangle$. The expected utility of each principal $j$ will be
\begin{equation}
    u^{j}(b, \phi; S) = \Big(v^{j} - \sum_{i \in [n]}t_{j \to i}(b) \Big) \cdot \delta^{j} \cdot f(S) \, . \label{eq: principal utility}
\end{equation}

Fix a principal $j$. Let $b$ and $\tilde{b}$ be two bid profiles
with $b^{j} = v^{j}$ and $b^{-j} = \tilde{b}^{-j}$. Let $S$ and $S'$ be the equilibria chosen by the agents in subgames $\mathcal{G} = \langle b, \phi\rangle$ and $\mathcal{G}' = \langle \tilde{b}, \phi\rangle$, respectively. The contract is said to be \emph{dominant-strategy incentive compatible} (\DSIC{}) if it always holds that
\begin{equation}
    u^{j}(b, \phi; S) \geq \max\Big\lbrace u^{j}(\tilde{b}, \phi; S'),\, 0 \Big\rbrace \, , \label{eq: DSIC}
\end{equation}
which means that a truthful bid is a dominant strategy and yields the principal $j$ non-negative utility. 

Note that this definition incorporates both \emph{truthfulness} and \emph{individual rationality} (IR).

\paragraph{Set functions and oracles.} For a subset $S \subseteq [n]$ and an element $i \in S$, define $\marginal{S}{i} = f(S) - f(\setminussingle{S}{i})$. The following classes of the set function $f(\cdot)$ are often considered in the literature.
\begin{itemize}
    \item \texttt{Additive}: $f$ is additive if for any $S \subset [n]$, $f(S) = \sum_{i \in S}f(\makeset{i})$.
    \item \texttt{Submodular}: $f$ is submodular if for any $S, S' \subseteq [n]$, $f(S) + f(S') \geq f(S \cup S') + f(S \cap S')$.
    \item \texttt{XOS}: $f$ is XOS if there exists $K$ additive functions $\lbrace a_{k} \rbrace_{k \in [K]}$ for some finite $K > 0$ such that for any $S \subseteq [n]$, $f(S) = \max_{k \in [K]}a_{k}(S)$.
    \item \texttt{Subadditive}: $f$ is subadditive if for any $S, S' \subseteq [n]$, $f(S) + f(S') \geq f(S \cup S')$.
\end{itemize}
It is well known that $\texttt{Additive} \subsetneq \texttt{Submodular} \subsetneq \texttt{XOS} \subsetneq \texttt{Subadditive}$. Throughout this paper, we assume that $f$ is \XOS{} and strictly monotone, that is, $f(S) < f(S')$ for any two subsets $S, S'$ such that $S \subsetneq S'$. 

This paper also adopts standard assumptions that the set function $f$ is succinctly represented, and the contract design algorithm has oracle access to $f$. There are two types of oracles often considered in the literature, the value oracle and the demand oracle.
\begin{itemize}
    \item Value oracle: when queried with a subset $S \subseteq [n]$, a value oracle returns the value $f(S)$.
    \item Demand oracle: when queried with an additive function $a: 2^{n} \mapsto \mathbb{R}_{\geq 0}$, an $\alpha$-demand oracle with $\alpha \in (0, 1]$ returns a subset $S \subseteq [n]$ satisfying $f(S) - a(S) \geq \alpha \cdot f(S') - a(S')$ for any $S' \subseteq [n]$.  An $\alpha$-demand oracle is also called a demand oracle when $\alpha = 1$. \cite{pmlr-v97-harshaw19a} show how to construct an $\alpha$-demand oracle in polynomial time for specific $\alpha$ values.
\end{itemize}

\paragraph{Value maximization and approximation.} Given an instance $\mathcal{I} = \langle m, n, f, v, c \rangle$ of the multi-agency model, the objective of this paper is to develop an algorithm that computes a \DSIC{} contract $\phi = \lbrace t, \delta \rbrace$ for $\mathcal{I}$ with a polynomial number of accesses to the oracles and maximizes 
\[
\big(\sum_{j \in [m]} v^{j} \delta^{j}\big) \cdot f\big( S(v, \phi) \big),\]
the expected total value of the eligible principals. We refer to this objective function as \emph{eligible values} for the rest of this paper. It can be viewed as a generalization of success probability maximization studied in the literature on contract design with a single principal (e.g.~\cite{FeldmanTzurPonitkaSchlesinger25, EC2025}).

For each instance $\mathcal{I}$ of the multi-agency model, let $\Phi(\mathcal{I})$ be the set of all \DSIC{} contracts, and $R: \mathbb{R}_{\geq 0}^{m} \mapsto \mathbb{R}_{\geq 0}$ be a \emph{revenue} function that maps the value profile $v$ to a non-negative real number. Define
\begin{equation}
    \texttt{OPT}(\mathcal{I}, R) = \sum_{j \in [m]} v^{j} \cdot \max_{\phi \in \Phi(\mathcal{I})}\Bigg\lbrace f\big( S(v, \phi) \big) \,\Bigg| \, \sum_{i \in [n], j \in [m]} t_{j \to i}(v) \leq R(v) \Bigg\rbrace \label{eq: definition of benchmark}
\end{equation}

be the upper bound on the \ev{} when the total transfer does not exceed $R(v)$. It can be viewed as the product of the total value of all principals (no one is excluded) and the project's optimal success probability by incentivizing the agents with budget $R(v)$ \cite{FeldmanTzurPonitkaSchlesinger25}. We call $R$ a revenue function because the total transfer that we collect from the principals resembles the revenue in auctions. We define the following notion with respect to the revenue function to evaluate the performance of contract design algorithms.

\begin{definition}[Approximation of Contract Design Algorithms]\label{def: algorithm approximation}
    A contract design algorithm is said to be $(\zeta,\, R)$-approximation with $\zeta \geq 1$ if for any instance $\mathcal{I}$, the algorithm generates a contract $\phi \in \Phi(\mathcal{I})$ such that 
    \begin{equation}
        \Big(\sum_{j \in [m]} v^{j} \delta^{j}\Big) \cdot f\big( S(v, \phi) \big) \geq \frac{1}{\zeta} \cdot \texttt{OPT}(\mathcal{I}, R) \, . \label{eq: approx ratio}
    \end{equation}
\end{definition}

In this paper, we examine the following categories of revenue functions in the approximation analysis.
\begin{enumerate}
    \item $\texttt{SumV}(v) := \sum_{j \in [m]}v^{j}$. 
    \item $\mathcal{F}_{x}^{2}(v) := x \cdot \max_{2 \leq j \leq m}\lbrace j \cdot \pi^{j}(v)\rbrace$, where $x>0$ and $\pi^{j}(v)$ is the $j$-th largest value in $v$. We set $\mathcal{F}_{x}^{2}(v)=0$ when $m<2$. 
    The notion $\mathcal{F}^{2}(v) \defeq \mathcal{F}_{x=1}^{2}(v)$ was suggested by~\cite{FiatGHK02} and in our context it represents the optimal fixed-price
revenue with at least two eligible principals.  
\end{enumerate}

\section{Broker-based Contract}\label{sec:broker}

Designing a contract for the multi-agency model requires the algorithm to make decisions over 
$O(m \cdot n)$-variables $\lbrace t_{j \to i}(b)\rbrace_{i \in [n], j \in [m]}$, and these decisions simultaneously affect the agents' responses in the subgames and the principals' strategies in the bidding game. This makes the design and analysis of the contract design algorithm complex.

To address this issue, this paper proposes a technical framework called \emph{BROker-based ContrAct DEsign} (BROCADE). In this framework, a non-strategic broker is introduced to collect transfers from the principals according to the bid profile and redistribute them to the agents to incentivize their actions. This mechanism effectively decouples the intricate interactions between principals and agents while reducing the number of decision variables to a linear scale of $O(m + n)$.

Denote the broker by $\beta$. The BROCADE protocol operates as follows.
\begin{enumerate}
    \item The broker declares two sets of transfer functions, $t_{\ANY \to \beta} = \lbrace t_{j \to \beta}: \mathbb{R}_{\geq 0}^{m} \mapsto \mathbb{R}_{\geq 0}\rbrace_{j \in [m]}$ and $t_{\beta \to \ANY} = \lbrace t_{\beta \to i}: \mathbb{R}_{\geq 0}^{m} \mapsto \mathbb{R}_{\geq 0} \rbrace_{i \in [n]}$, and a set of eligibility function, $\delta = \lbrace \delta^{j}: \mathbb{R}_{\geq 0}^{m} \mapsto  \lbrace 0, 1\rbrace \rbrace$. 
    For each bid profile $b \in \mathbb{R}_{\geq 0}^{m}$, $t_{j \to \beta}(b)$ and $t_{\beta \to i}(b)$ specify the transfer from principal $j$ to the broker and the transfer from the broker to agent $i$ when the project succeeds, respectively.
    Note that the broker also cannot observe agents' action profile; this is why the transfer functions depend only on the principal bids.
    \item Principals simultaneously report their bids $b$.
    \item Agents decide their response $S \subseteq [n]$ in the subgame induced by $b$, $t_{\beta}$, and $\delta$, where $t_{\beta} = \big\lbrace t_{\ANY \to \beta}, t_{\beta \to \ANY} \big\rbrace$.
    \item The project's outcome is realized according to the success probability $f(S)$. If the project succeeds, each principal $j \in [m]$ pays the broker $t_{j \to \beta}(b)$, and each agent $i \in [n]$ gets the payment $t_{\beta \to i}(b)$ from the broker.
\end{enumerate}

Under the BROCADE framework, the utilities of the agents and principals, $u_{i}$ and $u^{j}$, are similar to \EQref{eq: agent utility} and \EQref{eq: principal utility}, respectively. In particular, for a bid profile $b$ and a subset $S$ of agents that exerts effort,
\begin{equation}
    u_{i}(b, t_{\beta}; S) = \begin{cases}
        t_{\beta \to i}(b) \cdot f(S) - c_{i} & \text{if}\; i \in S \\
        t_{\beta \to i}(b) \cdot f(S) & \text{if}\; i \notin S 
    \end{cases} \; , \quad \text{and} \quad
    u^{j}(b, t_{\beta}; S) = \big(v^{j} - t_{j \to \beta}(b) \big) \cdot \delta^{j} f(S) \, . \label{eq: utilities with broker}
\end{equation}

\subsection{Equivalence between broker-free and broker-based contracts}

On top of its natural applied motivation to platform design, we next show that the presence of a broker is also useful to algorithm design: we show that the broker-free model (Section~\ref{sec:preliminaries}) and the broker-based model defined here are equivalent, in a formal way.

The BROCADE framework therefore 
provides a perspective on the problem that informs algorithm design without altering the problem's underlying structure. In particular, any broker-free contract $\phi$ can be expressed as a \emph{budget-balanced} broker-based contract $\phi_{\beta}$, and vice versa. 

\begin{definition}
A broker-based transfer function $t_{\beta}$ is said to be Budget-Balanced (\BB{}) if for any bid profile $b$,
\begin{equation}
    \sum_{j \in [m]}t_{j \to \beta}(b) = \sum_{i \in [n]}t_{\beta \to i}(b)\, , \label{eq: original definition of BB for det contract}
\end{equation}
which means that the total transfer received by the broker equals the total transfer he distributes.
\end{definition}

\begin{definition}\label{def: utility-wise equivalence}
    For a broker-free transfer function $t$ and a broker-based transfer function $t_{\beta}$, we say they are utility-wise equivalent if for any eligibility function $\delta$, any bid profile $b$, and any subset $S$ of agents, 
    \begin{equation}
        u_{i}(b, \lbrace t, \delta\rbrace; S) = u_{i}(b, \lbrace t_{\beta}, \delta\rbrace ; S) \quad \forall i \in [n]\, , \quad \text{and} \quad u^{j}(b, \lbrace t, \delta\rbrace; S) = u^{j}(b, \lbrace t_{\beta}, \delta \rbrace; S)  \quad \forall j \in [m] \, . \label{eq: transfer equivalence}
    \end{equation}
    The relation between $t$ and $t_{\beta}$ is denoted by $t \overset{u}{\sim} t_{\beta}$.
\end{definition}

The following lemma can be inferred from Definition \ref{def: utility-wise equivalence} and Eq.~\eqref{eq: utilities with broker}.

\begin{lemma}
\label{lemma: transfer equivalence}
For a broker-free transfer function $t$ and a broker-based transfer function $t_{\beta}$, if for any bid profile $b$, it holds that
\begin{equation}
    \sum_{i \in [n]}t_{j \to i}(b) = t_{j \to \beta}(b) \quad \forall j \in [m] \, , \quad \text{and} \quad \sum_{j \in [m]}t_{j \to i}(b) = t_{\beta \to i}(b) \quad \forall i \in [n] \, , \label{eq: equivalence between transfer functions}
\end{equation}
then we have $t \overset{u}{\sim} t_{\beta}$.
\end{lemma}

\begin{lemma}\label{lemma: contract equivalence}
For any broker-free transfer function $t$, there always exists a \BB{} broker-based transfer function $t_{\beta}$ such that $t \overset{u}{\sim} t_{\beta}$, and vice versa.
\end{lemma}

\begin{proof}
Given a broker-free transfer function $t$, a utility-wise equivalent broker-based transfer function $t_{\beta}$ can be constructed by following Eq.~\eqref{eq: equivalence between transfer functions}. 

It is \BB{} because
\begin{equation*}
    \sum_{j \in [m]}t_{j \to \beta}(b) = \sum_{j \in [m]}\sum_{i \in [n]}t_{j \to i}(b) = \sum_{i \in [n]} t_{\beta \to i}(b) \, .
\end{equation*}

Now, suppose we are given a \BB{} broker-based contract $t_{\beta}' $. An equivalent broker-free contract $t'$ can be constructed by solving a fractional bin-packing problem (FBPP) for every bid profile $b$, where each agent $i$ is viewed as a bin with capacity $t_{\beta \to i}' (b)$, each principal $j$ is viewed as a splittable item of size $t_{\beta \to j}' (b)$, and $t_{j \to i}' (b)$ corresponds to the part of item $j$ allocated to the bin $i$.

The fact that the transfer function $t_{\beta} '$ is \BB{} means that in the corresponding instance of FBPP, the total bin capacity equals the total item size. Thus, in an optimal solution to the FBPP instance, all items are fully assigned and all bins are fully filled. It implies that $t'$ satisfies $\sum_{i' \in [n]}t_{j \to i'}'(b) = t_{j \to \beta}'(b)$ for every $j \in [m]$ and $\sum_{j' \in [m]} t_{j' \to i}'(b) = t_{\beta \to i}'(b)$ for every $i \in [n]$. By Lemma \ref{lemma: transfer equivalence}, $t' \overset{u}{\sim} t_{\beta}'$.
\end{proof}

Combining Lemma \ref{lemma: contract equivalence} with Eq.~\eqref{eq: subgame equilibrium} and Eq.~\eqref{eq: DSIC} gives the following result.

\begin{theorem}\label{theorem: objective equivalence}
    For any \DSIC{} broker-free contract $\phi = \lbrace t, \delta \rbrace$ and any bid profile $b$, there exists a \DSIC{} \& \BB{} broker-based contract $\phi_{\beta} = \lbrace t_{\beta}, \delta \rbrace$ such that $f\big( S(b, \phi) \big) = f\big( S(b, \phi_{\beta}) \big)$.
\end{theorem}

Theorem \ref{theorem: objective equivalence} implies that the objective of this paper can be restated equivalently as follows: develop an algorithm that computes a \DSIC{} \& \BB{} broker-based contract $\phi_{\beta} = \lbrace t_{\beta}, \delta \rbrace$ for the given instance $\mathcal{I}$ with a polynomial number of oracle queries. This algorithm aims to maximize the \ev{}, $\sum_{j \in [m]}v^{j}\delta^{j}\cdot f\big( S(v, \phi_{\beta}) \big)$.

\section{Incentive Compatibility and Monotone Allocation Rules}\label{sec: myerson}

This section studies the \DSIC{} condition under the BROCADE framework. We begin with the agent transfers. For a target set $S\subseteq[n]$, the minimum success-contingent transfers that make $S$ an equilibrium are~\cite{DuttingEzraFeldmanKesselheim23}

\begin{equation}
    t_{\beta \to i} = \begin{cases}
        \frac{c_{i}}{\marginal{S}{i}} & \text{if} \; i \in S\\
        0 & \text{if} \; i \notin S
    \end{cases} \, . \label{eq: linear contract}
\end{equation}

Write $Q(S)=\sum_{i\in S}c_i/f(i:S)$, with $Q(\emptyset)=0$, and $T=\sum_j t_{j\to\beta}(b)$ for the minimum agent payments and total principal payments. For nonempty $S$, any surplus $T-Q(S)\geq0$ can be paid to an agent in $S$, chosen by a fixed order. The set $S$ remains an equilibrium, while every equilibrium is contained in $S$ because unpaid agents strictly prefer not to exert effort. This redistribution preserves the success probability and principal utilities and ensures exact budget balance. 
Consequently, the constructions below need only verify $T\geq Q(S)$ to establish \BB{}; we leave surplus redistribution implicit.

\paragraph{Allocation and payments.} A contract chooses the principals' eligibility profile $\delta$, the target agent set $S$, and the principal transfers $t_{\ANY\to\beta}$. We encode these choices as an allocation rule and a payment rule.

\begin{definition}[Allocation rule]
    Given a bid profile $b$, an allocation rule $\RULE{A}: \mathbb{R}_{\geq 0}^{m} \mapsto \lbrace 0, 1\rbrace^{m} \times2^{[n]}$ returns a pair of variables $\lbrace \delta, S\rbrace = \RULE{A}(b)$ that designate an eligibility profile for the principals and a subset of agents to incentivize, respectively.
\end{definition}

In the following, we abuse the notation and use $\RULE{A}(b; j)$ to represent the multiplication of the eligibility and the success probability for each principal $j$, that is,
\begin{equation}
    \RULE{A}(b; j) = \delta^{j} \cdot f(S) \, . \label{eq: allocation outcome for single principal}
\end{equation}

\begin{definition}[Payment rule]
    Given a bid profile $b$, a payment rule $\RULE{P}: \mathbb{R}_{\geq 0}^{m} \mapsto \mathbb{R}_{\geq 0}^{m}$ returns a vector $t_{\ANY \to \beta} = \RULE{P}(b)$ that specifies the transfer $t_{j \to \beta}(b)$ of each principal $j$ to the broker.
\end{definition}

With this transfer convention, a pair $\lbrace\RULE A,\RULE P\rbrace$ specifies a contract: $\RULE A$ determines $\delta$ and $S$, and $\RULE P$ determines $t_{\ANY\to\beta}$. Agents receive the minimum payments in Eq.~\eqref{eq: linear contract}, together with any surplus distributed as above. The characterization below therefore depends only on the effective allocation $\delta^j f(S)$ and the principal payments.

\begin{definition}[Implementability]\label{def: implementability}
    An allocation rule $\RULE{A}$ is implementable if there exists a payment rule $\RULE{P}$ such that the contract $\lbrace \RULE{A}, \RULE{P} \rbrace$ is \DSIC{}.
\end{definition}

Next, we will state the sufficient and necessary conditions for the implementability of the allocation rule.

\begin{definition}[Monotonicity]
    An allocation rule $\RULE{A}$ is said to be monotone if for any principal $j \in [m]$ and any two bid profiles $b$ and $\tilde{b}$ with $b^{j} \geq \tilde{b}^{j}$ and $b^{j'} = \tilde{b}^{j'}$ for each $j' \neq j$, $\RULE{A}(b; j) \geq \RULE{A}(\tilde{b}; j)$ always holds.
\end{definition}

The following theorem is inspired by Myerson's Lemma. However, Myerson's auction outputs only $\delta$, while here we also have $S$, which is a public good setting. In other words, in our case, all principals share the same success probability $f(S)$ while in Myerson's single-parameter setting, bidders may have individual winning probabilities. 
This, in turn, also implies a payment function -- different from Myerson's -- that we need to carefully construct.

\begin{theorem}\label{theorem: myerson}
    An allocation rule $\RULE{A}$ is implementable if and only if (iff.) it is monotone. For a monotone allocation rule $\RULE{A}$, the payment rule $\RULE{P}$ for which the contract $\lbrace \RULE{A}, \RULE{P} \rbrace$ is \DSIC{} can be formulated as
    \begin{equation}
        \RULE{P}^{j}(b) = \begin{cases}
        b^{j} - \frac{1}{\RULE{A}\big(b;\, j \big)}\int_{0}^{b^{j}}\RULE{A}\big(\tilde{b}^{j}, b^{-j};\, j \big)\, d\tilde{b}^{j} & \text{if} \; \RULE{A}\big(b;\, j \big) > 0 \\
        0 & \text{if} \; \RULE{A}\big(b;\, j \big) = 0
    \end{cases} \, . \label{eq: payment function}
    \end{equation}
\end{theorem}

\begin{proof}
Fix a principal $j \in [m]$ and the bids $b^{-j}$ of principals other than $j$. For any non-negative real number $x$, define $\gamma(x) = \RULE{A}\big(x, b^{-j};\, j \big)$, and $\eta(x) = \RULE{P}^{j}(x, b^{-j})$. By Eq.~\eqref{eq: utilities with broker}, the expected utility of principal $j$ when it has a value of $v^{j} = x \geq 0$ and submits a bid of $b^{j} = y \geq 0$ will be $(x - \eta(y)) \cdot \gamma(y)$. Then Eq.~\eqref{eq: DSIC} holds iff.~for any real numbers $x \geq y \geq 0$, the following three inequalities hold together.
\begin{align}
    (y - \eta(y)) \cdot \gamma(y) \; \geq& \; (y - \eta(x)) \cdot \gamma(x) \label{eq: submit higher no good} \\
    (x - \eta(x)) \cdot \gamma(x) \; \geq& \; (x - \eta(y)) \cdot \gamma(y) \label{eq: submit lower no good} \\
    (x - \eta(x)) \cdot \gamma(x) \; \geq& \; 0 \label{eq: truthful submit non-negative} \, .
\end{align}
In particular, Eq.~\eqref{eq: submit higher no good} means that if principal $j$ has a low value of $y$, then the utility by bidding truthfully dominates the utility by submitting a higher bid $x$. Eq.~\eqref{eq: submit lower no good} is for the case where principal $j$ has a high value of $x$ and is symmetric to Eq.~\eqref{eq: submit higher no good}. Eq.~\eqref{eq: truthful submit non-negative} is to ensure that the utility incurred by the truthful bid is always non-negative. By rearranging Eq.~\eqref{eq: submit higher no good} and \eqref{eq: submit lower no good}, we have
\begin{equation}
    y \cdot \gamma(y) - y \cdot \gamma(x) \; \geq \; \eta(y)\gamma(y) - \eta(x)\gamma(x) \; \geq\; x \cdot \gamma(y) - x\cdot \gamma(x) \, . \label{eq: sandwich condition}
\end{equation}

Note that Eq.~\eqref{eq: sandwich condition} requires $y \cdot \gamma(y) - y \cdot \gamma(x) \geq x \cdot \gamma(y) - x\cdot \gamma(x)$, which is equivalent to $(y - x)[\gamma(y) - \gamma(x)] \geq 0$. It holds only if $\gamma(y)\leq \gamma(x)$ for any $x \geq y$, which means that the monotonicity of the allocation rule is a necessary condition of \DSIC{}. To show that it is a sufficient condition, we proceed to prove that the payment function constructed in Eq.~\eqref{eq: payment function} satisfies Eq.~\eqref{eq: truthful submit non-negative} and \eqref{eq: sandwich condition} when $\RULE{A}$ is monotone. For any real number $z \geq 0$, we rewrite Eq.~\eqref{eq: payment function} as

\begin{equation}
    \eta(z) = \begin{cases}
        z - \frac{1}{\gamma(z)}\int_{0}^{z}\gamma(z')dz' & \text{if} \; \gamma(z) > 0 \\
        0 & \text{if} \; \gamma(z) = 0
    \end{cases} \, . \label{eq: payment function rewritten}
\end{equation}
The function $\eta$ in Eq.~\eqref{eq: payment function rewritten} is well-defined, because the function $\gamma$ is monotone, 
which implies that the integral $\int_{0}^{z}\gamma(z')dz'$ exists. The integral is also bounded, since $\gamma(z') = \RULE A(z', b^{-j};j) \in [0,1]$ for any $z' \geq 0$.

Nonnegative transfers and IR follow from $0\leq\eta(z)\leq z$. When $\gamma(z)>0$, monotonicity bounds the integral in Eq.~\eqref{eq: payment function rewritten} between $0$ and $z\gamma(z)$. When $\gamma(z)=0$, the definition gives $\eta(z)=0$. Thus Eq.~\eqref{eq: truthful submit non-negative} holds in both cases.

To prove that Eq.~\eqref{eq: sandwich condition} is satisfied for any $x \geq y \geq 0$, we need to consider three subcases.
\begin{enumerate}
    \item $\gamma(y) > 0$. Since $\gamma$ is monotonically increasing, in such a case, we have $\gamma(x) > 0$. Then by Eq.~\eqref{eq: payment function},
        \begin{equation*}
        \eta(y)\gamma(y) - \eta(x)\gamma(x) \; = \; y \cdot \gamma(y) - \int_{0}^{y}\gamma(z)dz \, - x \cdot \gamma(x) + \int_{0}^{x}\gamma(z)dz \; = \; y \cdot \gamma(y) - x \cdot \gamma(x) + \int_{y}^{x}\gamma(z)dz \, .
        \end{equation*}
    Still, by using the monotonicity of $\gamma$, we have
        \begin{equation*}
        (x - y) \cdot \gamma(y) \; =\; \int_{y}^{x}\gamma(y)dz \; \leq \; \int_{y}^{x}\gamma(z)dz \; \leq \; \int_{y}^{x}\gamma(x)dz \; =\; (x - y) \cdot \gamma(x)
        \end{equation*}
    Putting the two inequalities above together gives
        \begin{align*}
        \eta(y)\gamma(y) - \eta(x)\gamma(x) \; \geq& \; y \cdot \gamma(y) - x \cdot \gamma(x) + (x - y) \cdot \gamma(y) \; = \; x\cdot \gamma(y) - x \cdot \gamma(x) \, , \\
        \eta(y)\gamma(y) - \eta(x)\gamma(x) \; \leq& \; y \cdot \gamma(y) - x \cdot \gamma(x) + (x - y) \cdot \gamma(x) \; = \; y\cdot \gamma(y) - y \cdot \gamma(x) \, ,
        \end{align*}
    which establishes Eq.~\eqref{eq: sandwich condition}.
    \item $\gamma(x) = 0$. In such a case, $\gamma(y)$ also equals zero. Then Eq.~\eqref{eq: sandwich condition} trivially holds.
    \item $\gamma(x) > 0$ and $\gamma(y) = 0$. In such a case, Eq.~\eqref{eq: sandwich condition} is equivalent to
    \begin{equation*}
        y\; \leq\; \eta(x)\; \leq \; x \, .
    \end{equation*}
    We have proved that $x \geq \eta(x)$ always holds, and will prove $\eta(x) \geq y$ by using the monotonicity of $\gamma$. In particular, since $\gamma(y) = 0$, we have $\gamma(y') = 0$ for any $y' \in [0, y]$. Therefore, 
    \begin{equation*}
        \int_{0}^{x}\gamma(x')dx' \; = \; \int_{y}^{x}\gamma(x')dx' \; \leq \; \int_{y}^{x}\gamma(x)dx' \; = \; \gamma(x) \cdot (x - y) \, .
    \end{equation*}
    Therefore,
    \begin{equation*}
        \eta(x) \; = \; x - \frac{1}{\gamma(x)}\int_{0}^{x}\gamma(x')dx' \; \geq \; x - \frac{1}{\gamma(x)}\cdot \gamma(x) \cdot (x - y) \; = \; y\, ,
    \end{equation*}
    which completes the proof.
\end{enumerate}

\end{proof}

\section{Budget-Scalable Selection}\label{sec: selection}

For an allocation rule $\RULE{A}$ defined in Section \ref{sec: myerson}, it has to make decisions on both sides of the contract entities, that is, specifying the eligibility $\delta$ for the principals and selecting the subset $S$ of agents to exert effort. This section breaks down the allocation decisions by defining a \emph{selection} function $\SEL: \mathbb{R}_{\geq 0} \mapsto 2^{[n]}$. Given a non-negative real number $B \geq 0$, called the \emph{budget}, as the input, a selection function $\SEL$ outputs a subset $S = \SEL(B) \subseteq [n]$ of agents that satisfies the following \emph{budget constraint}.
\begin{equation}
    \sum_{i \in S}\frac{c_{i}}{f(i: S)} \; \leq \; B \, . \label{eq: budget constraint}
\end{equation}
The selection function will be combined with different eligibility determination algorithms, either deterministic or random, in the next two sections, to generate desired allocation rules. The following notions and propositions are used in the design of the selection function and the analysis of the generated allocation rule and contract.

\begin{definition}[Budget-scalability]
    A selection rule $\SEL$ is budget-scalable if for any two real numbers $B \geq B' \geq 0$, it holds that 
    \begin{equation}
        f\big( \SEL(B) \big) \; \geq \; f\big( \SEL(B') \big) \, , \quad \forall B, B' \in \mathbb{R}_{\geq 0},\, B\geq B' \, . \label{eq: budget scalability}
\end{equation}
\end{definition}

For a selection function $\SEL$, let $\gamma_{f, \, \SEL}: \mathbb{R}_{\geq 0} \mapsto [0, 1]$ be a scalar function such that for any $B \geq 0$, $\gamma_{f, \, \SEL}(B) = f\big(\SEL(B)\big)$. We may omit the subscript when $f$ and $\SEL$ are obvious from the context.

\begin{lemma}
    When $\SEL$ is budget-scalable, the function $\gamma$ is a monotonically non-decreasing step function.
\end{lemma}

\begin{proof}
    The function $\gamma$ is monotonically non-decreasing by Eq.~\eqref{eq: budget scalability}. It is a step function because it is non-decreasing, and its image is a finite set, as the image of $\SEL$ is a subset of $2^{[n]}$.
\end{proof}

\begin{definition}[Jump discontinuity]
    A positive real number $z > 0$ is said to be a jump discontinuity of the step function $\gamma(\cdot)$ if $\lim_{\epsilon \to 0}\gamma(z - \epsilon) \neq \lim_{\epsilon > 0}\gamma(z + \epsilon)$. 
\end{definition}

\begin{definition}[Optimal selection]\label{def: optimal selection}
    A selection function $\SEL$ is said to be optimal if for any budget $B \geq 0$, we have
    \begin{equation*}
        f\big( \SEL(B)\big) = \max_{S\subseteq [n]}\Big\lbrace f(S) \Big| \sum_{i \in S}\frac{c_{i}}{f(i:S)}\leq B \Big\rbrace \, .
    \end{equation*}
    Throughout, we denote the optimal selection function by $\SEL^{*}$.
\end{definition}

\begin{definition}[Approximate selection]
    A selection function $\SEL$ is said to be $\theta$-approximate for some $\theta \geq 1$ if for any budget $B \geq 0$, it holds that $f\big( \SEL^{*}(B) \big) \leq \theta \cdot f\big( \SEL(B) \big)$.
\end{definition}

It can be inferred from the optimality that $\SEL^{*}$ is a budget-scalable selection function. However, it cannot be adopted in our contract design due to computational infeasibility. In particular, even when $f$ is an additive function, the problem of optimizing the success probability $f$ under the budget constraint Eq.~\eqref{eq: budget constraint} is equivalent to the knapsack problem, which is NP-hard to solve \cite{FeldmanTzurPonitkaSchlesinger25}.
In \cite{EC2025}, Aharoni et al.~propose an $O(1)$-approximate algorithm for budget-constrained success probability maximization when $f$ is non-additive.

\begin{lemma}[\cite{EC2025}]\label{lemma: set-value oracle}
    Given a budget $B \geq 0$, there exists an algorithm $\mathtt{AHT}$ that outputs a subset $S \subseteq [n]$ of agents that satisfies Eq.~\eqref{eq: budget constraint} with respect to $B$. It runs in polynomial time using the value oracle and $\alpha$-demand oracle, and guarantees that $f\big(\SEL^{*}(B)\big) \leq \Big(\frac{48}{\alpha^2} + o(1) \Big)\cdot f(S)$.
For $B=0$, the algorithm returns $\emptyset$.
\end{lemma}

\begin{remark}
When the success probability function $f \in \mathtt{Submodular}$, an $\alpha$-demand oracle can be constructed using a polynomial number of queries to a value oracle \cite{harshaw2019submodular}, where $\alpha=1-\frac{1}{e}$.
\end{remark}

Although Lemma \ref{lemma: set-value oracle} provides an efficient algorithm for computing a constant approximation of $\SEL^{*}(B)$, it cannot be directly adopted as our selection function since we do not know whether \AHT{} is budget-scalable. Therefore, in the remainder of this section, we use it in a black-box manner to develop an $O(1)$-approximate budget-scalable selection function.

\subsection{Tagged partition-based selection}

We start by proving a property of the transfer $\frac{c_{i}}{f(i:S)}$ specified by Eq.~\eqref{eq: linear contract} to an agent $i$ for incentivizing an agent subset $S \ni i$.

\begin{lemma}\label{lemma: transfer for incentivizing a single agent}
    For any success probability function $f \in \mathtt{Subadditive}$, any non-empty agent subset $S \subseteq [n]$, and any agent $i \in S$, we have $\frac{c_{i}}{f(\makeset{i})} \leq \frac{c_{i}}{f(i:S)}$.\footnote{Since $\XOS{} \subseteq \texttt{Subadditive}$, this property directly holds for the \XOS{} success probability considered in this paper.}
\end{lemma}

\begin{proof}
    For any $f \in \mathtt{Subadditive}$, we have $f(\makeset{i}) + f(\setminussingle{S}{i}) \geq f(S)$, which implies that $f(\makeset{i}) \geq f(S) - f(\setminussingle{S}{i}) = f(i:S)$. Therefore, this proposition holds.
\end{proof}

Define $B_{\min} = \min_{i \in [n]}\frac{c_{i}}{f(\makeset{i})}$ and $B_{\max} = \sum_{i \in [n]}\frac{c_{i}}{f(i:[n])}$ to be the minimum budget to incentivize a single agent and the budget sufficient for incentivizing all the agents, respectively. The values of these two parameters can be obtained in $O(n)$-time using the value oracle. By Lemma \ref{lemma: transfer for incentivizing a single agent}, $B_{\min}$ is also the minimum budget for incentivizing any non-empty subset of agents, and $0<B_{\min} \leq B_{\max}$. Although for some success probability function $f \in \mathtt{XOS} \setminus \mathtt{Submodular}$, there may exist a subset $S \subsetneq [n]$ such that $\sum_{i \in S}\frac{c_{i}}{f(i:S)} > B_{\max}$, the set $[n]$ is always a no-worse option than such a subset $S$ for the selection function, as $f$ is assumed to be strictly monotone. Therefore, the design of the selection function only needs to focus on the interval $[B_{\min}, B_{\max}]$.

\begin{algorithm}[t]
\caption{Tagged Partition-based Selection}
\label{algo: tps}
    \SetAlgoLined
    \KwIn{$B \in \mathbb{R}_{\geq 0}$}
	\KwOut{$S \subseteq [n]$}
    $\langle \lambda, \mu \rangle \leftarrow \texttt{TaggedPartition}(n, f, c)$;\\
    $K \leftarrow |\lambda|$;\\
    $S \leftarrow \emptyset$;\\
    \For{$k \in [K]$}{
        \If{$\lambda_{k} > B$}{\textbf{break};}
        \If{$f(\mu_{k}) > f(S)$}{$S \leftarrow \mu_{k}$;}
    }
    \textbf{return} $S$;
\end{algorithm}

Our selection function is implemented as the algorithm \emph{Tagged Partition-based Selection} (\texttt{TPS}) shown in Algorithm \ref{algo: tps}. It first computes a pair $\langle \lambda, \mu\rangle$ using the number $n$ of agents, the success probability $f$, and the cost profile $c$ as input parameters. Here $\lambda$ is an increasing sequence of $K$ real numbers $\lbrace \lambda_1 = B_{\min}, \lambda_{2}, \cdots, \lambda_{K} = B_{\max} \rbrace$ that specifies a partition of the interval $[B_{\min}, B_{\max}]$, and for every $k \in [K]$, the tag $\mu_{k}$ is referred to as a tag of $\lambda_{k}$ and represents a subset of agents such that $\sum_{i \in \mu_{k}}\frac{c_{i}}{f(i: \mu_{k})} \leq \lambda_{k}$ and $\mu_{K} = [n]$. For analysis, set $(\lambda_0,\mu_0)=(0,\emptyset)$; this auxiliary pair is not stored by the algorithm. The same tagged partition, including the choices among tied outputs of the black box, is used for every bid profile. The computation of $\mu_{k}$ will be discussed in the next subsection. For any given budget $B \in [B_{\min}, +\infty)$, \texttt{TPS} chooses the agent subset
\begin{equation}
    S \; =\; \argmax_{\mu_{k}: \lambda_{k} \leq B} f(\mu_{k} ) \, , \label{eq: chosen tag}
\end{equation}
the tag that maximizes the success probability with the budget $\lambda_{k} \leq B$. The following proposition can be derived from the description of \texttt{TPS}.

\begin{lemma}\label{lemma: basic properties of tps}
    \texttt{TPS} satisfies Eq.~\eqref{eq: budget constraint} for any $B \geq 0$ and is budget-scalable.
\end{lemma}

\begin{proof}
By its definition, \texttt{TPS} satisfies Eq.~\eqref{eq: budget constraint} for any $B \geq 0$. Since $\langle \lambda, \mu\rangle$ does not depend on the input $B$, Eq.~\eqref{eq: chosen tag} implies that \texttt{TPS} is budget-scalable.    
\end{proof}

\begin{definition}[Jump ratio]
    A tagged partition $\langle \lambda, \mu \rangle$ is said to have a jump ratio of $r \geq 1$ if for any $k \in [0, K-1]$ and any budget $B \in [\lambda_{k},\, \lambda_{k+1})$,
\begin{equation}
    f(\mu_{k}) \; \geq \; \frac{1}{r}f\big( \SEL^*(B) \big) \, . \label{eq: jump ratio}
\end{equation}
\end{definition}

\begin{lemma}\label{lemma: from jump ratio to approx factor}
    When provided with a tagged partition $\langle \lambda, \mu \rangle$ with jump ratio $r$, \texttt{TPS} serves as an $r$-approximate selection function.
\end{lemma}

\begin{proof}
For $B<B_{\min}$, no nonempty agent set is feasible, so \texttt{TPS} and the optimum both return the empty set. For $B\geq B_{\max}$, the tag $\mu_K=[n]$ is affordable and optimal. Otherwise, let $k=\max\{h:\lambda_h\leq B\}$ be the last affordable threshold. Then $k<K$ and $B\in[\lambda_k,\lambda_{k+1})$. By the jump-ratio condition, $f(\mu_k)\geq f(\SEL^*(B))/r$. Equation~\eqref{eq: chosen tag} maximizes over every affordable tag, so its output $S$ satisfies $f(S)\geq f(\mu_k)$. The returned set need not be the tag at the last affordable threshold.
\end{proof}

\subsection{Computing the Tagged partition}

\begin{algorithm}[t]
\caption{Function \texttt{TaggedPartition}}
\label{algo: sad}
    \SetAlgoLined
    \KwIn{$n \in \mathbb{Z}_{\geq 1},\, f \in \XOS,\,  c \in \mathbb{R}_{> 0}^{n}$}
	\KwOut{$\langle \lambda, \mu \rangle \in \big( \mathbb{R}_{> 0} \times 2^{[n]} \big)$}
    $\lambda \leftarrow \emptyset,\; \mu \leftarrow \emptyset$;\\
    \For{$i \in [n]$}{
        $\rho_{i} \leftarrow \frac{c_{i}}{f(\makeset{i})}$;
    }
    $\mathbf{sort}([n], \rho, \text{ascending})$;\\
    $\rho_{n + 1} \leftarrow B_{\max}$;\\
    $k \leftarrow 1$;\\
    \For{$i \in [n]$}{
        $B' \leftarrow \rho_{i}$;\\
        \While{$B' < \rho_{i + 1}$}{
            $\lambda_{k} \leftarrow \makeset{B'}$;\\
            $\mu_{k} \leftarrow \makeset{\AHT(B')}$;\\
            $k \leftarrow k + 1$;\\
            $B' \leftarrow B' \cdot 2$;
        }
    }
    $\lambda_{k} \leftarrow \makeset{B_{\max}}$;\\
    $\mu_{k} \leftarrow [n]$;\\
    $\mathbf{return}\;\langle\lambda, \mu\rangle$;
\end{algorithm}

Next, we show how to implement the function \texttt{TaggedPartition} invoked in Algorithm \ref{algo: tps} to generate $\langle \lambda, \mu \rangle$ with well-bounded jump ratio $r$. For each agent $i \in [n]$, define $\rho_{i} = \frac{c_{i}}{f(\makeset{i})}$. As shown in Algorithm \ref{algo: sad}, \texttt{TaggedPartition} first sorts the elements in $[n]$ in ascending order of $\rho$ and takes $\rho_{n + 1} = B_{\max}$, which means that the agents are relabeled. Then, for every agent $i \in [n]$, \texttt{TaggedPartition} performs a doubling iteration starting from $B' = \rho_{i}$. In each round of the iteration, \texttt{TaggedPartition} adds the pair $\langle B', \AHT(B')\rangle$ to $\langle \lambda, \mu \rangle$, where $\AHT(B')$ represents the agent subset chosen by the algorithm in Lemma \ref{lemma: set-value oracle} with respect to $B'$, and then doubles $B'$ until $B' \geq \rho_{i + 1}$. The last step is to add $\langle B_{\max}, [n]\rangle$ to $\langle \lambda, \mu \rangle$ and output $\langle \lambda, \mu \rangle$. The elements of $\lambda$ and $\mu$ are indexed according to their insertion order. The following results are derived from the description of \texttt{TaggedPartition}.

\begin{lemma}\label{lemma: full interval}
    For any budget $B \in [B_{\min}, B_{\max}]$, we have $\lambda_{1} \leq B \leq \lambda_{K}$.
\end{lemma}

\begin{proof}
The sorted singleton thresholds satisfy $\rho_1=B_{\min}$. Repeated singleton thresholds produce empty loops. If some $\rho_i<B_{\max}$, the first nonempty loop starts at the same minimum value and inserts it as $\lambda_1$. If all thresholds equal $B_{\max}$, the final step inserts this value, which is also $B_{\min}$. The last inserted threshold is always $B_{\max}$, proving the claim.
\end{proof}

\begin{lemma}\label{lemma: doubling jumping}
    For any $k \in [K - 1]$, $\lambda_{k + 1} \leq 2 \cdot \lambda_{k}$.
\end{lemma}

\begin{proof}
    Suppose that $\lambda_{k}$ is added in the iteration starting from $\rho_{i}$. When $\lambda_{k} \cdot 2 < \rho_{i + 1}$, $\lambda_{k + 1}$ is defined inside the ``while'' loop, which implies the claim. When $\lambda_{k} \cdot 2 \geq \rho_{i + 1}$,
    then $\lambda_{k+1} = \rho_{i + 1}$ and the claim again holds.
    
\end{proof}

The following lemma is inferred from Lemma \ref{lemma: full interval} and Lemma \ref{lemma: doubling jumping}.

\begin{lemma}\label{lemma: double increment}
    For any budget $B \in [B_{\min}, B_{\max}]$, there exists $k \in [K]$ such that $\lambda_{k} \leq B \leq 2 \cdot \lambda_{k}$.
\end{lemma}

After analyzing the structure of $\lambda$, we proceed to investigate the property of value maximization over the subset of $[n]$. 

\begin{lemma}\label{lemma: budget infeasibility}
    For any $i \in [n]$ and any budget $B < \rho_{i}$, a subset $S \subseteq [n]$ of agents cannot be incentivized with $B$ if $S \cap \lbrace i, i + 1, \cdots ,n\rbrace \neq \emptyset$.
\end{lemma}

\begin{proof}
    Let the element with the largest index in $S$ be $\hat{i}$. The transfer needed to incentivize $S$ is
    \begin{equation*}
        \sum_{i' \in S} \frac{c_{i'}}{f(i': S)}\; \geq \; \frac{c_{\hat{i}}}{f(\hat{i}: S)} \; \geq \; \frac{c_{\hat{i}}}{f(\makeset{\hat{i}})} \; \geq \; \frac{c_{i}}{f(\makeset{i})} \; > \; B \, ,
    \end{equation*}
    where the second inequality holds because $f$ is an \XOS{}  function, which implies that $f(\hat{i}: S) = f(S) - f(\setminussingle{S}{\hat{i}}) \leq f(\makeset{\hat{i}})$.
\end{proof}

For any budget $B$ and any subset $S \subseteq [n]$, define $\gamma^{*}(B; S) = \max_{S' \subseteq S: \sum_{i \in S'}c_{i}/f(i: S') \leq B}f(S')$, which represents the maximum value of incentivizing the subset of $S$ with budget $B$. The following proposition can be directly inferred from Lemma \ref{lemma: budget infeasibility}.

\begin{lemma}\label{lemma: local to global optimality}
    For any $i \in [n]$ and any budget $B < \rho_{i}$, $\gamma^{*}(B; [i - 1]) = \gamma^{*}(B; [n])$.
\end{lemma}

To analyze the jump ratio of the tagged partition $\langle \lambda, \mu \rangle$ produced by Algorithm \ref{algo: sad}, we utilize the following result proposed in \cite{FeldmanTzurPonitkaSchlesinger25} on optimizing an XOS function under the budget constraint.

\begin{lemma}[\cite{FeldmanTzurPonitkaSchlesinger25}]\label{lemma: doubling high payment}
    When $f$ is an \XOS{} function, for any subset $S$ of agents and any budgets $B, B'$ with $B > B' \geq \max_{i \in S}\frac{c_{i}}{f(\makeset{i})}$, we have
    \begin{equation*}
        \frac{\gamma^{*}(B; S)}{\gamma^{*}(B'; S)} \leq 64 \cdot \frac{B}{B'} + 1 \, .
    \end{equation*}
\end{lemma}

\begin{theorem}\label{theorem: sad}
    Algorithm \ref{algo: sad} produces a tagged partition $\langle \lambda, \mu \rangle$ of $[B_{\min}, B_{\max}]$ in polynomial time using value-oracle access and the $\alpha$-demand oracle. It ensures that the length $K$ of $\langle \lambda, \mu \rangle$ is bounded by $O\bigg( n + \log{\Big(\frac{B_{\max}}{B_{\min}}\Big)} \bigg)$, and the jump ratio $r$ is bounded by a constant depending on $\alpha$. 
\end{theorem}

\begin{proof}
    By the description of \texttt{TaggedPartition}, the length is bounded by
    \begin{equation*}
        K \; \leq \; \sum_{i \in [n]} \Big\lceil \log_{2}\frac{\rho_{i + 1}}{\rho_{i}} \Big\rceil + 1 \; \leq \; \sum_{i \in [n]} \Big(\log_{2}\frac{\rho_{i + 1}}{\rho_{i}} + 1 \Big) + 1 \; = \; \log_{2}\frac{\rho_{n + 1}}{\rho_{1}} + n + 1\, .
    \end{equation*}
    For each tag, the algorithm described in Lemma \ref{lemma: set-value oracle} is invoked only once, which ensures the oracle efficiency of \texttt{TaggedPartition}.

    Let $r_\alpha=48/\alpha^2+o(1)$ be the approximation factor in Lemma~\ref{lemma: set-value oracle}. Fix $B\in[B_{\min},B_{\max})$, let $i$ be the largest index with $\rho_i\leq B$, and choose the last affordable threshold $\lambda_k$. The construction within $[\rho_i,\rho_{i+1})$ gives $\rho_i\leq\lambda_k\leq B<2\lambda_k$. By Lemma~\ref{lemma: budget infeasibility}, optimizing at either $B$ or $\lambda_k$ can be restricted to agents in $[i]$. Also $\lambda_k\geq\max_{h\leq i}\rho_h$, as required by the price-frugality lemma. If $B=\lambda_k$, the guarantee of \AHT{} directly gives the required approximation. For $B>\lambda_k$, we have
    \begin{align*}
        f\big(S^{*}(B) \big) \; =& \; \gamma^{*}(B; [n]) \; = \; \gamma^{*}(B; [i]) \; \leq \; \Big(64 \cdot \frac{B}{\lambda_{k}} + 1\Big) \gamma^{*}(\lambda_{k}; [i]) \\
        \leq& \; 129 \cdot \gamma^{*}(\lambda_{k}; [i]) \; = 129 \cdot \gamma^{*}(\lambda_{k}; [n]) \;
        \leq \; 129 \cdot r_{\alpha} \cdot f\big( \AHT(\lambda_{k})\big) \; = \; 129 r_{\alpha} \cdot f(\mu_{k}) \,
    \end{align*}
    where the second transition is by Lemma \ref{lemma: local to global optimality}, the third one follows from Lemma \ref{lemma: doubling high payment}, the fourth one uses $B<2\lambda_k$, and the sixth one is inferred from Lemma \ref{lemma: set-value oracle}. By Eq.~\eqref{eq: jump ratio}, the jump ratio $r$ is bounded by $129r_{\alpha}$, which completes the proof. Budgets below $B_{\min}$ have only the empty feasible set, while $\mu_K=[n]$ is optimal at and above $B_{\max}$.
\end{proof}

\begin{remark}
The construction in Algorithm~\ref{algo: sad} and Theorem~\ref{theorem: sad} is a black-box construction that does not rely on any properties of algorithm \AHT{}. Specifically, Algorithm~\ref{algo: sad} can adopt any $r$-approximate algorithm for the success probability maximization problem under the budget constraint and transform it to a budget-scalable $(129 \cdot r)$-approximate selection function.

\end{remark}

To summarize, \begin{corollary}
Algorithm \ref{algo: tps} provides a budget-scalable selection function which is $O(1)$-approximate.
\end{corollary}

\section{Deterministic Contracts: Positive and Negative Results}\label{sec: non-excludability}

Having laid the technical foundation in previous sections,
this section discusses deterministic contracts. We first present a non-ex\-cludable deterministic \DSIC{} contract that has a constant approximation of $\mathtt{OPT}(\cdot, \mathtt{SumV})$, the best possible benchmark for any \DSIC{} contract. However, it may inherently violate the \BB{} condition, as we establish an impossibility: no deterministic \DSIC{} + \BB{}
 contract provides a finite approximation against $\mathtt{OPT}(\cdot, R)$ for any $R$ with $R(v) \geq \max_{j}v^{j}$ over all possible $v \in \mathbb{R}_{\geq 0}^{m}$. 

\subsection{A Non-excludable Deterministic Contract}

In this subsection, we design a \emph{Non-ExcludablE Deterministic} (\texttt{NEED})

contract using the selection function \texttt{TPS} designed in Section \ref{sec: selection}. The contract is said to be non-excludable as its allocation rule outputs a constant eligibility profile $\delta = \vec{1}$ for any bid profile $b$.

Since the eligibility profile is fixed and the payment function is given in Theorem \ref{theorem: myerson}, it remains to decide the budget $B$ for the selection function $\texttt{TPS}$. However, it is not straightforward to determine an appropriate budget. 
Recall that the budget constraint on the selection function requires that
\begin{equation*}
    B \; \geq \; \sum_{i \in S} \frac{c_{i}}{f(i: S)}\, ,
\end{equation*}
where $S$ is the selected subset of agents. Here $B$ is the budget supplied to agent selection; it need not equal the principal payments. Budget balance would additionally require those payments to cover the agent transfers. Meanwhile, the payment function given by Theorem \ref{theorem: myerson} is formulated as
\begin{equation}
    \RULE{P}^{j}(b) = \begin{cases}
        b^{j} - \frac{1}{\RULE{A}(b;\, j)} \int_{0}^{b^j} \RULE{A}(\tilde{b}, b^{-j};\, j)  d\tilde{b} & \text{if} \;  \RULE{A}(b;\, j)  > 0 \\
        0 & \text{otherwise}
    \end{cases} \, , \nonumber
\end{equation}
where $\RULE{A}(b; \, j)$ relies on both the eligibility profile and the selected subset of agents.
Since principal payments depend on the allocation, we instead specify the agent-selection budget independently, without imposing budget balance. We introduce an efficiently computable function, nondecreasing in every coordinate, $\mathbf{B} : \mathbb{R}_{\geq 0}^{m} \mapsto \mathbb{R}_{\geq 0}$ of the bid profile to convert the original budget constraint to the following form. 
\begin{equation}
    \mathbf{B}(b) \geq \sum_{i \in S} \frac{c_{i}}{f(i:S)} \, . \label{eq: relaxed feasibility}
\end{equation}
Once the value of $\mathbf{B}(b)$ can be determined independently of the payment rule, we can decouple the computation of the allocation rule from that of the payment rule. The expression of $\mathbf{B}$ will be specified at the end of this section. 

With the function $\textbf{B}$, the contract \texttt{NEED} $\phi = \lbrace \RULE{A}, \RULE{P} \rbrace$ can be stated as follows. Given a bid profile $b$, the allocation rule $\RULE{A}(b) = \Big\lbrace \vec{1}, \texttt{TPS}\big( \mathbf{B}(b) \big) \Big\rbrace$ specifies an eligibility profile $\vec{1}$ and an agent subset chosen by \texttt{TPS} with respect to the budget $\mathbf{B}(b)$. The payment by each principal $j \in [m]$ is computed by plugging
\begin{equation}
    \RULE{A}(b; j) \; = \; \delta^{j} \cdot f\Big( \texttt{TPS}\big( \mathbf{B}(b) \big) \Big) \; = \; f\Big( \texttt{TPS}\big( \mathbf{B}(b) \big) \Big) \label{eq: deterministic allocation}
\end{equation}
into Eq.~\eqref{eq: payment function} as follows.
\begin{equation}
        \RULE{P}^{j}(b) = \begin{cases}
        b^{j} - \frac{1}{f\Big( \texttt{TPS}\big( \mathbf{B}(b) \big) \Big)}\int_{0}^{b^{j}}f\Big( \texttt{TPS}\big( \mathbf{B}(\tilde{b}, b^{-j}) \big) \Big)\, d\tilde{b} & \text{if} \; \mathbf{B}(b) \geq B_{\min} \\
        0 & \text{if} \; \mathbf{B}(b) < B_{\min}
    \end{cases} \, . \label{eq: deterministic payment}
\end{equation}

The following lemma shows that, by allowing the the payment rule to be fully aware of the tagged partition $\langle \lambda, \mu \rangle$ used by the selection function \texttt{TPS}, the integral in Eq.~\eqref{eq: deterministic payment} can be computed efficiently.

\begin{definition}[Breakpoint Oracle]
A breakpoint oracle for a coordinatewise nondecreasing budget function $\mathbf B$ returns $x(z)=\inf\{x\geq0:\mathbf B(x,b^{-j})\geq z\}$, given $b^{-j}$ and $z\geq0$. We set $\inf\emptyset=+\infty$.
\end{definition}

\begin{remark}
For $\mathbf B=\mathtt{SumV}$, the breakpoint is $x(z)=\max\{0,z-\sum_{\ell\ne j}b^\ell\}$. For $\mathbf B=\mathcal F_a^2$ with fixed $a>0$, sort the other bids as $d_1\geq\cdots\geq d_{m-1}$. If $z/a\leq\mathcal F^2(b^{-j})$, then $x(z)=0$. Otherwise, $x(z)$ is the minimum of $z/(ak)$ over $2\leq k\leq m$ satisfying $d_{k-1}\geq z/(ak)$, with the minimum of an empty set equal to $+\infty$. Indeed, when the other bids alone cannot attain the target, a supporting group of size $k$ must include $j$ and $k-1$ other bids of at least $z/(ak)$. Sorting takes $O(m\log m)$ time once per principal, after which each query takes $O(m)$ time.
\end{remark}

\begin{lemma}\label{lemma: payment complexity of NEED}
Given a breakpoint oracle, the payment in Eq.~\eqref{eq: deterministic payment} is computable with $O(K)$ breakpoint queries, value queries, and arithmetic operations, where $K=O(n+\log(B_{\max}/B_{\min}))$.
\end{lemma}

\begin{proof}
Fix $j$ and $b^{-j}$. Precompute $M_k=\max_{0\leq h\leq k}f(\mu_h)$ for $k=0,\ldots,K$, using $\mu_0=\emptyset$ and $M_0=0$. For every $k\in[K]$, let $x_k=\min\{b^j,x(\lambda_k)\}$, and set $x_0=0$ and $x_{K+1}=b^j$. These points are nondecreasing, but need not be distinct. They partition the bid interval, so
    \begin{equation*}
        \int_{0}^{b^{j}}f\Big( \texttt{TPS}\big( \mathbf{B}(\tilde{b}, b^{-j}) \big) \Big)\, d\tilde{b} =  \sum_{k=0}^{K} \int_{x_k}^{x_{k+1}}f\Big( \texttt{TPS}\big( \mathbf{B}(\tilde{b}, b^{-j}) \big) \Big)\, d\tilde{b}.
    \end{equation*}
For each nonempty open interval $(x_k,x_{k+1})$, exactly the first $k$ tag thresholds are affordable. Monotonicity of $\mathbf B$ ensures this even when a breakpoint is not attained: a point strictly above the infimum belongs to the upward-closed feasible set. Thus the success probability on that interval is $M_k$. Repeated points contribute zero, and values at finitely many endpoints do not affect the integral. Therefore,
    \begin{equation*}
        \int_{0}^{b^{j}}f\Big( \texttt{TPS}\big( \mathbf{B}(\tilde{b}, b^{-j}) \big) \Big)\, d\tilde{b} = \sum_{k=0}^{K}M_k(x_{k+1}-x_k).
    \end{equation*}
Clipping unattainable breakpoints to $b^j$ and retaining all zero breakpoints handles budgets that skip tags or already afford tags at bid zero. The formula also covers $b^j=0$. When the effective allocation at the submitted bid is zero, the payment is zero without division; otherwise substitute the integral into Eq.~\eqref{eq: deterministic payment}. Computing the prefix maxima and scanning the $K$ thresholds proves the claimed bound.
\end{proof}

\begin{lemma}\label{lemma: monotonicity of NEED}
    The allocation rule $\RULE{A}$ of the contract \texttt{NEED} is monotone.
\end{lemma}

\begin{proof}
    Fix a principal $j \in [m]$. Let $b$ and $\tilde{b}$ be two bid profiles with $b^{j} \geq \tilde{b}^{j}$ and $b^{-j} = \tilde{b}^{-j}$. Then we have
    \begin{equation*}
        \RULE{A}(b; j) \; = \; f\Big( \texttt{TPS}\big( \mathbf{B}(b) \big) \Big) \; \geq \; f\Big( \texttt{TPS}\big( \mathbf{B}(\tilde{b}) \big) \Big) = \RULE{A}(\tilde{b}; j) \, ,
    \end{equation*}
    where the second transition holds because $\mathbf{B}$ is monotonically non-decreasing and \texttt{TPS} is budget-scalable by Lemma \ref{lemma: basic properties of tps}.
\end{proof}

Next, we present the approximation guarantee of \texttt{NEED} (recall Definition \ref{def: algorithm approximation}).

\begin{lemma}\label{lemma: approx ratio of NEED}
Let $r$ be the jump ratio of \texttt{TPS}. If $R$ is efficiently computable, nondecreasing in every coordinate, and admits an efficient breakpoint oracle, then \texttt{NEED} with $\mathbf B=R$ is a \DSIC{} $(r,R)$-approximate contract.
\end{lemma}

\begin{proof}
    Let $B=R(v)$. If $B<B_{\min}$, only the empty set is feasible and the benchmark is zero. Otherwise let $k=\max\{h:\lambda_h\leq B\}$ be the last affordable threshold. Consider the following two cases.
    \begin{itemize}
        \item $k < K$. By the definition of jump ratio, we have $f(\mu_{k}) \geq \frac{1}{r}f\big( \SEL^{*}(B) \big)$.
        \item $k = K$. In such a case, we have $f(\mu_{k}) = f([n]) = f\big( \SEL^{*}(B) \big)$.
    \end{itemize}
    Therefore, with a truthful bid $b = v$, we have
    \begin{equation*}
        \Big( \sum_{j \in [m]} v^{j}\delta^{j} \Big) \cdot f\Big( \texttt{TPS}\big( \mathbf{B}(b) \big) \Big) = f\Big( \texttt{TPS}\big( R(v) \big) \Big) \cdot \sum_{j \in [m]} v^{j} \geq f(\mu_{k}) \cdot \sum_{j \in [m]} v^{j} \geq \frac{1}{r} f\Big( \SEL^{*}\big( R(v) \big) \Big) \cdot \sum_{j \in [m]} v^{j} \, ,
    \end{equation*}
    where the first equality holds because $\delta = \vec{1}$. 
    By Eq.~\eqref{eq: definition of benchmark} and Definition \ref{def: optimal selection}, we have $\texttt{OPT}(\mathcal{I}, R) = f\Big( \SEL^{*}\big( R(v) \big) \Big) \cdot \sum_{j \in [m]} v^{j}$. Therefore, this proposition follows from Eq.~\eqref{eq: approx ratio}. Monotonicity and the payment characterization give DSIC, and Lemma~\ref{lemma: payment complexity of NEED} gives efficient implementation under the stated oracle assumptions.
\end{proof}

Putting Theorem \ref{theorem: myerson}, Lemma \ref{lemma: payment complexity of NEED}, Lemma \ref{lemma: monotonicity of NEED}, and Lemma \ref{lemma: approx ratio of NEED} together, we get the following result.

\begin{theorem}\label{theorem: main result for NEED contract}
    \texttt{NEED} can be implemented as an $\big( O_{\alpha}(1), \texttt{SumV} \big)$-approximate \DSIC{} contract for any problem instance $\mathcal{I}$ in polynomial time using the value oracle and $\alpha$-demand oracle, where $O_{\alpha}(1)$ is a constant depending on $\alpha$.
\end{theorem}

For example, let $m=2$, $n=1$, and $c_1=f(\{1\})=1$. At bids $(1,1)$, \texttt{NEED} selects the agent. Each principal's effective allocation remains one even when she lowers her bid to zero, so both integral payments are zero, while the agent requires a transfer of one.

\subsection{Impossibility Results for Deterministic Contracts}

Note that the contract described in Theorem \ref{theorem: main result for NEED contract} may not be \BB{}. This is because the approximation guarantee assumes the allocation rule has a budget $B = \sum_{j \in [m]} v^{j}$, while the total transfer collected from the principals can be strictly lower than the total value. However, the impossibility result given in Theorem \ref{theorem: approximation impossbility with DSIC and BB} implies that to achieve an approximation guarantee $(\zeta, \mathtt{SumV})$ for any finite real number $\zeta \geq 1$, it is unavoidable to violate the \BB{} constraint.

\begin{theorem}\label{theorem: approximation impossbility with DSIC and BB}
    Fix $m\geq1$ and a revenue function $R$ satisfying $R(v)\geq\max_{j\in[m]}v^j$ for every $v\in\mathbb R_{\geq0}^m$. For every finite $\zeta\geq1$, no deterministic \DSIC{} and \BB{} contract guarantees $(\zeta,R)$-approximation over all instances, even with two agents, strictly positive costs, and an additive success probability function.\footnote{
    Note that for the existing results on contract design with a single principal, the DSIC condition is not defined in the absence of bids. Therefore, the impossibility result stated in Theorem \ref{theorem: approximation impossbility with DSIC and BB} does not conflict with the approximation guarantees in the previous literature.}
\end{theorem}

\begin{remark}
    This theorem implies that no $(\zeta, \mathtt{SumV})$-approxi\-mate contract can be both \DSIC{} and \BB{} for any finite $\zeta$, since $\mathtt{SumV}(v) \geq \max_{j \in [m]}v^{j}$ always holds.
\end{remark}

\begin{proof}

Fix $\zeta$ and $R$. Choose $B>0$ and $z>\max\{2,\zeta+1\}$. With two agents, define the additive function $f$ by $f(\{1\})=1/z$ and $f(\{2\})=(z-1)/z$, and set $c_1=B/z^2$ and $c_2=B(z-1)/z$. Incentivizing $\{1\}$, $\{2\}$, and $\{1,2\}$ requires total transfers of at least $B/z$, $B$, and $B+B/z$, respectively. Consider the value profiles $v(x)=(x,\vec0)$ at $x=B/z$ and $x=B$, with corresponding instances $\mathcal I_L$ and $\mathcal I_H$.

Fix any deterministic \DSIC{} and \BB{} contract on these public parameters. Whenever the success probability is positive, the zero-valued principals pay zero by IR and the eligibility convention. Thus any positive eligible value is funded by principal~1 alone, whose payment is at most $x$ by IR.

At $v(B/z)$, the benchmark is at least $B/z^2>0$, since $R(v(B/z))\geq B/z$ can fund agent~1. If the contract produces zero eligible value there, it already fails on $\mathcal I_L$. Otherwise it must admit principal~1 and incentivize exactly agent~1; BB and IR force her payment to be $B/z$.

Now consider $v(B)$. If the contract admits principal~1 and incentivizes agent~2, BB and IR force her payment to be $B$, giving her zero utility. Reporting $B/z$ instead yields utility $(B-B/z)/z>0$, contradicting DSIC. Hence its eligible value at $v(B)$ is at most $B/z$. Since $R(v(B))\geq B$ can fund agent~2, the benchmark is at least $B(z-1)/z$, and
\begin{equation*}
    \frac{B\delta^1 f(S)}{\mathtt{OPT}(\mathcal I_H,R)}\leq\frac{1}{z-1}<\frac{1}{\zeta}.
\end{equation*}
Thus every such contract fails on at least one of $\mathcal I_L$ and $\mathcal I_H$; the failing value profile may depend on the contract.

\end{proof}

\section{The Sampling Toll Sharing Contract}
\label{sec:satosh}

\subsection{Definitions and Extensions}

This section studies random contracts. We first extend contract equivalence and the monotonicity characterization to the randomized setting. We then prove an $\Omega(n)$ approximation lower bound under strong benchmarks, motivating a weaker benchmark budget under which \SC{} achieves an $O(\log m)$ guarantee independent of $n$.

\paragraph{Random contracts.} A random contract $\phi$ is defined to be a distribution $\mathcal{D}$ over deterministic contracts. Throughout this paper, we assume that $\mathcal{D}$ is always a discrete distribution. The distribution is chosen independently of the bids. Each deterministic contract $\bar{\phi} = \lbrace \bar{t}, \bar{\delta} \rbrace \in \operatorname{supp}(\mathcal{D})$ corresponds to a possible \emph{realization} of the coins used by $\phi$. 
For a random contract $\phi$ with distribution $\mathcal{D}$, the objective is still to maximize the \ev{}, which is formulated as $\mathbb{E}_{\bar{\phi} \sim \mathcal{D}}\Big[ \sum_{j \in [m]}v^{j} \cdot \bar{\delta}^{j}(b) \cdot f\big( S(b, \bar{\phi}) \big) \Big]$. The notation $S(b, \bar{\phi})$ in this expression is defined in Section \ref{sec:preliminaries} and represents the subgame equilibrium of agents in $\mathcal{G} = \langle b, \bar{\phi} \rangle$. For random contracts, the approximation guarantee in Definition~\ref{def: algorithm approximation} is evaluated in expectation over the contract's coins.

\begin{definition}[Universally \DSIC{} Contract]\label{definition: universally DSIC}
    A random contract $\phi$ is \emph{universally} \DSIC{} if for every realization of the contract's coins, Eq.~\eqref{eq: DSIC} is satisfied by the corresponding deterministic contract $\bar{\phi}$.
\end{definition}

\begin{definition}[Universal \BB{}]\label{definition: universally BB}
    A random contract $\phi$ is universally \BB{} if Eq.~\eqref{eq: original definition of BB for det contract} holds under every realization of the coins used by $\phi$.
    
\end{definition}

We aim to develop a random contract that is both universally \DSIC{} and universally \BB{} and maximizes the \ev{}.

\paragraph{Contract equivalence.} Definitions \ref{definition: universally DSIC} and \ref{definition: universally BB} imply that a random contract that is both universally \DSIC{} and universally \BB{}  can be viewed as a distribution $\mathcal{D}$ over a set of deterministic contracts that are \DSIC{} and \BB{}. Therefore, one can prove the following result on contract equivalence, which generalizes Theorem \ref{theorem: objective equivalence} to the random setting.

\begin{theorem}\label{theorem: random contract equivalence}
    For any broker-free random contract $\phi$ that is both universally \DSIC{} and universally \BB{}, there exists a broker-based random contract $\phi_{\beta}$ which is also universally \DSIC{} and universally \BB{} and has the same \ev{} as $\phi$ for any bid profile $b$.
\end{theorem}

\begin{proof}
    By Lemma \ref{lemma: contract equivalence} and Theorem \ref{theorem: objective equivalence}, one can define a mapping $\mathcal{T}$ such that for each broker-free deterministic \DSIC{} \& \BB{} contract $\bar{\phi} = \lbrace \bar{t}, \bar{\delta} \rbrace$, $\mathcal{T}(\bar{\phi}) = \bar{\phi}_{\beta} = \lbrace \bar{t}_{\beta}, \bar{\delta} \rbrace$ specifies a broker-based deterministic \DSIC{} \& \BB{} contract with $f\Big(S\big(b, \bar{\phi}\big) \Big) = f\Big(S\big(b, \mathcal{T}(\bar{\phi})\big) \Big)$ for every bid profile $b$.

    For an arbitrary broker-free random contract $\phi$ with a distribution $\mathcal{D}$, we define $\phi_{\beta}$ to be a broker-based random contract with a distribution $\mathcal{D}_{\beta}$ such that for every broker-based deterministic contract $\bar{\phi}_{\beta}$,
    \begin{equation*}
        \mathcal{D}_{\beta}(\bar{\phi}_{\beta}) \; = \; \sum\limits_{\bar{\phi} \in \operatorname{supp}(\mathcal{D}): \: \mathcal{T}(\bar{\phi}) = \bar{\phi}_{\beta}} \mathcal{D}(\bar{\phi}) \, .
    \end{equation*}

    It can be inferred from Definitions \ref{definition: universally DSIC} and \ref{definition: universally BB} that $\phi_{\beta}$ is universally \DSIC{} and universally \BB{} if and only if $\phi$ has the same properties. Furthermore, for every bid profile $b$,
    \begin{align*}
        \sum_{\bar{\phi} \in \operatorname{supp}(\mathcal{D})} \mathcal{D}(\bar{\phi})\sum_{j \in [m]}v^{j} \bar{\delta}^{j}(b) \cdot f\big(S(b, \bar{\phi}) \big) \; =& \; \sum_{\bar{\phi} \in \operatorname{supp}(\mathcal{D})} \mathcal{D}(\bar{\phi})\sum_{j \in [m]}v^{j} \bar{\delta}^{j}(b) \cdot f\Big(S\big(b, \mathcal{T}(\bar{\phi}) \big) \Big) \\
        =& \; \sum_{\bar{\phi}_{\beta} \in \operatorname{supp}(\mathcal{D_{\beta}})} \mathcal{D}_{\beta}(\bar{\phi}_{\beta})\sum_{j \in [m]}v^{j} \bar{\delta}^{j}(b) \cdot f\Big(S\big(b, \bar{\phi}_{\beta} \big) \Big) \, ,
    \end{align*}
    which implies that $\phi$ and $\phi_{\beta}$ have the same \ev{}.
\end{proof}

\paragraph{Allocation rule and payment rule.} Theorem \ref{theorem: random contract equivalence} allows us to focus on broker-based random contracts, and express a random contract as the combination of the allocation rule and the payment rule. In particular, consider an arbitrary random contract $\phi$ with a distribution $\mathcal{D}$. For each deterministic contract $\bar{\phi} \in \operatorname{supp}(\mathcal{D})$, let $\RULE{A}_{\bar{\phi}}$ and $\RULE{P}_{\bar{\phi}}$ denote its allocation and payment rules, respectively. We define the random allocation rule $\RULE{A}$ and payment rule $\RULE{P}$ of $\phi$ jointly by
\begin{equation*}
    \big(\RULE{A}(b), \RULE{P}(b) \big) \; \defeq \;  \big(\RULE{A}_{\bar{\phi}}(b), \RULE{P}_{\bar{\phi}}(b) \big)\, , \qquad \bar{\phi} \sim \mathcal{D}
\end{equation*}
where a single realization of $\bar{\phi}$ is used by both $\RULE{A}$ and $\RULE{P}$.

\begin{definition}[Universal monotonicity]\label{def: universal monotonicity}
    An allocation rule $\RULE{A}$ of a random contract $\phi$ with distribution $\mathcal{D}$ is universally monotone if for every $\bar{\phi} \in \operatorname{supp}(\mathcal{D})$, its allocation rule $\RULE{A}_{\bar{\phi}}$ is monotone.
\end{definition}

The following propositions can be inferred from Definition \ref{def: implementability}, Theorem \ref{theorem: myerson}, and Definition \ref{definition: universally DSIC}.

\begin{corollary}\label{corollary: randomized myerson}
    A random allocation rule $\RULE{A}$ is implementable by a universally \DSIC{} contract if and only if it is universally monotone.
\end{corollary}

\begin{corollary}\label{corollary: random myerson payment}
    For a universally monotone random allocation rule $\RULE{A}$, applying Eq.~\eqref{eq: payment function} to each deterministic realization defines a payment rule $\RULE{P}$ such that $\lbrace\RULE{A},\RULE{P}\rbrace$ is universally \DSIC{}.
\end{corollary}

\paragraph{Approximation lower bound.} For random contracts, we establish a lower bound that grows with the number of agents. The proof compares several value profiles under the same distribution of deterministic realizations.

\begin{theorem}\label{theorem: impossiblity for universally truthful}
    Fix $m\geq1$, $n\geq2$, and a revenue function $R$ satisfying $R(v)\geq\max_{j\in[m]}v^j$ for every $v\in\mathbb R_{\geq0}^m$. Every universally \DSIC{} and universally \BB{} random contract has an instance $\mathcal I$ with $m$ principals, $n$ agents, strictly positive costs, and an additive success probability function on which its expected eligible value is less than $\frac{2}{n}\mathtt{OPT}(\mathcal I,R)$. Thus the worst-case approximation ratio is $\Omega(n)$, even for any fixed $m$.
\end{theorem}

\begin{proof}

Let $a=n+1$, $w_i=a^{i-1}/\sum_{h=1}^n a^{h-1}$, and $c_i=iw_i$ for each agent $i\in[n]$. Define $f(S)=\sum_{i\in S}w_i$. This additive function is strictly monotone, $f([n])=1$, and all costs are positive. Incentivizing agent $i$ requires a transfer of at least $i$, independently of the other agents' actions. The geometric weights satisfy $\sum_{h<i}w_h<w_i/n$.

Fix an arbitrary universally \DSIC{} and universally \BB{} random contract on these public parameters. Its distribution over deterministic realizations is the same for all value profiles considered below. For each $i\in[n]$, let $v(i)=(i,\vec0)$ and write $\mathcal I_i$ for the corresponding instance. Since $R(v(i))\geq i$ can fund agent $i$, we have $\mathtt{OPT}(\mathcal I_i,R)\geq iw_i$.

Let $E_i$ be the event that the realized contract admits principal~1 and incentivizes agent $i$ at $v(i)$. On $E_i$, the zero-valued principals pay zero by IR and the eligibility convention. BB and IR therefore force principal~1 to pay exactly $i$. No other agent can be incentivized within this total payment, so the success probability is $w_i$.

The events $E_i$ are pairwise disjoint. Indeed, if a deterministic realization satisfies both $E_i$ and $E_k$ with $i<k$, principal~1 with value $k$ has zero utility when bidding truthfully, but obtains $(k-i)w_i>0$ by reporting $i$. This contradicts DSIC for that realization. Hence $\sum_{i=1}^n\Pr(E_i)\leq1$, and some $i$ satisfies $\Pr(E_i)\leq1/n$.

At this profile $v(i)$, put $q_i=\delta^1(v(i))f(S(v(i)))$. If $q_i>0$, principal~1 is the only source of payments and pays at most $i$, so no agent with index greater than $i$ can be incentivized. Outside $E_i$, it follows that $q_i\leq\sum_{h<i}w_h$. Thus
\begin{equation*}
    \mathbb E[q_i]\leq w_i\Pr(E_i)+\sum_{h<i}w_h<\frac{2w_i}{n}.
\end{equation*}
The expected eligible value is $i\mathbb E[q_i]<2iw_i/n\leq(2/n)\mathtt{OPT}(\mathcal I_i,R)$, as claimed. The index $i$ is chosen from the distribution of the contract, independently of the realized coins. Since $m$ was arbitrary, the lower bound holds even when the number of principals is fixed.

\end{proof}

Theorem~\ref{theorem: impossiblity for universally truthful} shows that, when the benchmark budget is always at least the maximum principal value, the approximation ratio must grow as $\Omega(n)$ even for fixed $m$. Thus these benchmarks preclude a guarantee depending only on the number of principals. To obtain such a guarantee, we use the weaker budget $\mathcal F_{1/4}^2(v)$, a quarter of the optimal fixed-price revenue supported by at least two principals. This budget need not track an isolated high value.\footnote{For example, if $v=(1,0)$, the lower bound assumes $R(v)\geq1$, whereas $\mathcal F_{1/4}^2(v)=0$.} With this benchmark, \SC{} is universally \DSIC{} and universally \BB{} and achieves an $O_\alpha(\log m)$ approximation for $m\geq2$, independent of the number of agents; for $m=1$ the benchmark is zero.

\subsection{Algorithm Description}

Our random algorithm, \emph{SAmpling TOll SHaring} (\SC{}), is inspired by the auction mechanism proposed in \cite{FiatGHK02}, while the presence of the success probability function leads to several differences in ensuring the universally \DSIC{} property. It relies on a terminology called \emph{optimal fixed-price revenue} and a subprocedure to share a given toll\footnote{To avoid ambiguity, we do not use the terminology \emph{cost} for the principals (bidders) like \cite{FiatGHK02}, because in the current paper, this terminology is used to represent the loss of the agents when exerting effort.} $\tau$ over a subset $X \subseteq [m]$ of principals with a fixed payment.

\begin{definition}[Optimal fixed-price revenue \cite{FiatGHK02}]\label{definition: optimal fixed price revenue}
    Given a bid profile $b$, its optimal fixed-price revenue is $\mathcal{F}(b) = \max_{k \in [|b|]} k \cdot \pi^{k}(b)$, where $\pi^{k}(b)$ is the $k$-th largest bid in $b$.
\end{definition}

\paragraph{Consistent ordering of equal revenues.}
Following the ordering convention of~\cite{FiatGHK02}, fix a distinct priority
$\sigma^j=j$ for each principal, independently of the bids and the random
partition. Higher priorities break ties. For real numbers $a,a'$ and
priorities $s,s'$, write
\[
 (a,s)\succ_{\mathrm{lex}}(a',s')
 \quad\Longleftrightarrow\quad
 a>a'\ \text{or}\ (a=a'\ \text{and}\ s>s').
\]
For a nonempty group $G$, let $G[\ell]$ denote the principal at rank $\ell$
in decreasing lexicographic order of $(b^j,\sigma^j)$. Define its
\emph{ordered revenue} as the pair
\[
 \widehat{\mathcal F}(b^G)
 =\max_{\mathrm{lex},\,1\leq k\leq |G|}
       \big(kb^{G[k]},\sigma^{G[k]}\big).
\]
Its first coordinate is the fixed-price revenue $\mathcal F(b^G)$ from
Definition~\ref{definition: optimal fixed price revenue}. For an empty group,
set $\mathcal F(b^\emptyset)=0$ and
$\widehat{\mathcal F}(b^\emptyset)=(0,-\infty)$.

\begin{definition}[Toll sharing with consistent ordering
\cite{FiatGHK02}]\label{definition: toll sharing procedure}
Given the labeled bids $b^G$ and a target pair
$\widehat{\tau}=(\tau,\sigma_*)$, define
\[
\texttt{TollSharing}(b^G,\widehat{\tau})
 =\max\Big(\{0\}\cup\big\{k\in[|G|]:
(kb^{G[k]},\sigma^{G[k]})
       \succ_{\mathrm{lex}}\widehat{\tau}\big\}\Big).
\]
If the returned value is $k>0$, accept the first $k$ principals in this
ordering. Each accepted principal pays the share $\tau/k$ upon success.
If $k=0$, no principal is accepted and no payment is charged.
\end{definition}

The next two lemmas establish monotonicity. Every comparison of labeled
bids uses the fixed priorities defined above.

\begin{lemma}\label{lemma: monotone fixed-price revenue}
For a fixed group $G$, raising any principal's bid cannot decrease either
$\mathcal F(b^G)$ or $\widehat{\mathcal F}(b^G)$, where the latter is compared
lexicographically.
\end{lemma}

\begin{proof}
Raising one labeled bid cannot decrease any order statistic of the pairs
$(b^j,\sigma^j)$. 

Consequently, the maximum over the pairs
$(kb^{G[k]},\sigma^{G[k]})$ cannot decrease. The same argument applied to
the first coordinates proves the assertion for $\mathcal F$.
\end{proof}

\begin{lemma}\label{lemma: monotonicity of toll sharing}
Let $b^G$ and $\tilde b^G$ be bids of the same labeled principals, with
$b^j\geq\tilde b^j$ for every $j\in G$. For any fixed target
$\widehat{\tau}$,
\[
 \texttt{TollSharing}(b^G,\widehat{\tau})
 \geq \texttt{TollSharing}(\tilde b^G,\widehat{\tau}).
\]
\end{lemma}

\begin{proof}
Let $\tilde k$ be the value returned on $\tilde b^G$. If $\tilde k=0$, the
claim is immediate. Otherwise, its rank-$\tilde k$ revenue pair strictly
exceeds $\widehat{\tau}$. Increasing bids cannot decrease that pair, so
$\tilde k$ remains feasible on $b^G$. The largest feasible rank on $b^G$ is
therefore at least $\tilde k$.
\end{proof}

\begin{algorithm}[t]
\color{black}
\SetAlFnt{\small\color{black}}
\SetAlCapFnt{\small\color{black}}
\SetAlCapNameFnt{\small\color{black}}
\caption{\SC{} with consistent revenue ordering}
\label{algo: sampling}
    \SetAlgoLined
    \KwIn{$\mathcal{I}=\langle m,n,f,b,c\rangle$; fixed priorities $\sigma^j=j$}
    \KwOut{$\delta\in\{0,1\}^m$, $S\subseteq[n]$, $\RULE P\in\mathbb R_{\geq0}^m$}
    Partition $[m]$ into $X,Y$ using one independent fair coin per principal\;
    $\widehat{\tau}^X\leftarrow\widehat{\mathcal F}(b^X)$\;
    $\widehat{\tau}^Y\leftarrow\widehat{\mathcal F}(b^Y)$\;
    \eIf{$\widehat{\tau}^X\succ_{\mathrm{lex}}\widehat{\tau}^Y$}{
        $W\leftarrow X$; $\widehat{\tau}\leftarrow\widehat{\tau}^Y$\;
    }{
        $W\leftarrow Y$; $\widehat{\tau}\leftarrow\widehat{\tau}^X$\;
    }
    $(\tau,\sigma_*)\leftarrow\widehat{\tau}$\;
    $k\leftarrow\texttt{TollSharing}(b^W,\widehat{\tau})$\;
    $\delta\leftarrow\vec 0$; $\RULE P\leftarrow\vec 0$\;
    \For{$\ell=1,\ldots,k$}{
        $j\leftarrow W[\ell]$\;
        $\delta^j\leftarrow1$; $\RULE P^j\leftarrow\tau/k$\;
    }
    $S\leftarrow\texttt{TPS}(\tau)$\;
    \textbf{return} $\delta,S,\RULE P$\;
\end{algorithm}

Algorithm~\ref{algo: sampling} uses the same random partition for its allocation
and payment rules. Write
$\widehat{\tau}^G=(\tau^G,\sigma_G)$ for $G\in\{X,Y\}$.
Since $m\geq1$, at least one group is nonempty. Nonempty disjoint groups have
distinct revenue priorities, so the ordered revenues select a unique,
nonempty winning group $W$. Its maximum revenue pair strictly exceeds the
other group's target pair. Hence some positive rank is feasible and the
algorithm always has $k\geq1$.

The toll supplied to \texttt{TPS} as the budget is
\[
 \tau=\min\{\tau^X,\tau^Y\}
      =\min\{\mathcal F(b^X),\mathcal F(b^Y)\}.
\]
Each of the first $k$ principals in the fixed-priority ordering of $W$ is eligible and pays $\tau / k$.

\begin{lemma}
\label{lemma:satosh-bb}
    \SC{} generates a universally \BB{} contract.
\end{lemma}
\begin{proof}
For any fixed realization of the coins, exactly $k$ principals each pay $\tau/k$, so total principal transfers equal $\tau$. If $S=\texttt{TPS}(\tau)$ is nonempty, budget feasibility gives $Q(S)\leq\tau$, and the redistribution convention in Section~\ref{sec: myerson} makes budget balance exact while preserving the selected equilibrium. If $S=\emptyset$, then $\tau<B_{\min}$. Assign the full transfer $\tau$ to fixed agent~1. Every other agent strictly prefers no effort, and agent~1 also prefers no effort because $\tau f(\{1\})<c_1$. Thus the empty set remains the unique equilibrium and total agent transfers again equal $\tau$. This includes $\tau=0$.
\end{proof}

In the next two subsections, we further investigate the properties of \SC{}, including the \DSIC{} property and the approximation guarantee, and prove the following statement.

\begin{theorem}
    For $m\geq2$, \SC{} provides a $\big(O_{\alpha}(\log m), \mathcal{F}_{1/4}^{2}\big)$-approximate contract that is both universally \DSIC{} and universally \BB{}. It runs in polynomial time using value and $\alpha$-demand oracles. For $m=1$, the benchmark is zero.
\end{theorem}

This result is immediately derived from Lemma~\ref{lemma:satosh-bb}, Corollary~\ref{corollary:satosh-dsic}, and Theorem~\ref{theorem: approx guarantee of satoshi}.

\subsection{Analysis of the \DSIC{} Property of \SC{}}

\SC{} builds on the sampling and toll-sharing structure of~\cite{FiatGHK02}. In our setting, the toll also finances the agents' effort and therefore affects the project's success probability. A principal's effective allocation is $\delta^j f(S)$, so the truthfulness analysis must account for how both eligibility and success probability respond to her bid.

Fix a realization of the random partition. The budget-scalable selection function \texttt{TPS} of Section~\ref{sec: selection} makes the success probability nondecreasing in the toll. To establish monotonicity of the effective allocation, we therefore show that raising a bid cannot decrease the toll and that an eligible principal remains eligible when she raises her bid.

The eligibility argument uses the consistent ordering in Algorithm~\ref{algo: sampling}. Fixed priorities resolve ties between group revenues, and the same lexicographic order determines eligibility in \texttt{TollSharing}. We first establish monotonicity of the toll and the resulting success probability, then verify that this ordering preserves eligibility. Finally, we verify the payment rule using Eq.~\eqref{eq: payment function}.

\begin{lemma}\label{lemma: universally monotone toll}
    Fix the realization of the coins. Let $\tau$ be the extracted toll corresponding to a bid profile $b$. When a principal $j \in [m]$ raises her bid from $b^{j}$ to $\hat{b}^{j} \geq b^{j}$, the resulting toll $\tau'$ satisfies $\tau' \geq \tau$.
\end{lemma}

\begin{proof}
For every bid profile, the toll is
\[
 \tau=\min\{\mathcal F(b^X),\mathcal F(b^Y)\},
\]
even when the two revenues are equal. Raising a bid affects only its own group's revenue,
which cannot decrease by Lemma~\ref{lemma: monotone fixed-price revenue};
the other group's revenue is unchanged. Taking the minimum therefore cannot
decrease the toll. Priorities resolve ties without changing either group's revenue.
\end{proof}

The next proposition follows from Eq.~\eqref{eq: chosen tag} and Lemma \ref{lemma: universally monotone toll}. 

\begin{lemma}\label{lemma: universally monotone success probability}
    Fix the coin realization arbitrarily. Let $S$ and $\hat{S}$ be the chosen subsets of agents corresponding to bid profiles $b$ and $\lbrace \hat{b}^{j}, b^{-j}\rbrace$ with $\hat{b}^{j} \geq b^{j}$, respectively. Then $f(S) \leq f(\hat{S})$.
\end{lemma}

\begin{lemma}\label{lemma: universally monotone eligibility}
    Under an arbitrarily fixed coin realization, if a principal $j$ satisfies $\delta^{j} = 1$, then raising her bid $b^{j}$ to $\hat{b}^{j} \geq b^{j}$ does not change her eligibility.
\end{lemma}

\begin{proof}
Fix the random partition, and suppose that $j$ is eligible in the winning
group $W$. Let $L$ be the other group, and let
$\widehat{\tau}=(\tau,\sigma_*)=\widehat{\mathcal F}(b^L)$.
Raising $j$'s bid leaves this target unchanged and cannot decrease
$\widehat{\mathcal F}(b^W)$ by
Lemma~\ref{lemma: monotone fixed-price revenue}. Thus $W$ remains the winning
group, including when the group revenues are equal.

Let $k$ and $\hat k$ be the old and new values of
$\texttt{TollSharing}(b^W,\widehat{\tau})$. By
Lemma~\ref{lemma: monotonicity of toll sharing}, $\hat k\geq k$.
Increasing $j$'s bid cannot worsen her rank in the fixed-priority ordering.
She was among the first $k$ principals before the increase, so she remains
among the first $\hat k$ afterward. Hence her eligibility is preserved.
\end{proof}

Putting Definition \ref{def: universal monotonicity}, Lemma \ref{lemma: universally monotone success probability} and Lemma \ref{lemma: universally monotone eligibility} together, we obtain the following result.

\begin{theorem}\label{theorem: universally monotone allocation rule}
    \SC{}'s allocation rule is universally monotone.
\end{theorem}

\begin{theorem}\label{theorem: equivalence of myerson and monopolistic pricing}
For every fixed realization of the random partition and every principal $j$,
the payment in Algorithm~\ref{algo: sampling} agrees with
Eq.~\eqref{eq: payment function} whenever $\RULE A(b;j)>0$.
For all bid profiles, including those with $\RULE A(b;j)=0$, its expected
payment $\RULE P^j(b)\RULE A(b;j)$ agrees with the expected payment induced
by that equation.
\end{theorem}

\begin{proof}
Fix the random partition, a principal $j$, and all other bids $b^{-j}$.
First suppose $\RULE A(b;j)>0$. Let $G$ be $j$'s group, $L$ the other
group, and $k$ the current number of eligible principals. Write
\[
 (\tau,\sigma_*)=\widehat{\mathcal F}(b^L),
 \qquad p=\frac{\tau}{k}.
\]
Whenever $j$ is eligible, $G$ wins and its budget is the fixed value $\tau$.
Thus her effective allocation is the constant
$q=f(\texttt{TPS}(\tau))>0$ whenever she is eligible. Define her critical bid by
\[
 x=\inf\{z\geq0:\RULE A(z,b^{-j};j)>0\}.
\]
We show that $x=p$.

A rank $r$ is feasible precisely when at least $r$ principals $i\in G$ satisfy
\[
 (b^i,\sigma^i)\succ_{\mathrm{lex}}(\tau/r,\sigma_*).
\]
At her current bid, $j$ satisfies this inequality for $r=k$, and hence for
all $r>k$. None of these larger ranks is feasible, since $k$ is maximal.
Because $j$ already satisfies their inequalities, changing only her bid
cannot make any rank $r>k$ feasible.

If $j$ reports $z>p$, she and the other $k-1$ currently eligible principals
still satisfy the inequality for rank $k$. Thus $k$ remains the largest
feasible rank and $G$ wins. Principal $j$ must be among the first $k$:
otherwise, those $k$ principals together with $j$ would make rank $k+1$
feasible, a contradiction. Hence she is eligible and pays $p$.
Conversely, if $z<p$, eligibility at any rank $r\leq k$ would require
$z\geq\tau/r\geq p$, a contradiction.

Therefore $x=p$, with effective allocation zero below $x$ and equal to $q$
above $x$. The value at $x$ itself does not affect the integral in
Eq.~\eqref{eq: payment function}, which gives, at the current eligible bid,
\[
 b^j-\frac{1}{q}\int_0^{b^j}\RULE A(z,b^{-j};j)\,dz
=b^j-(b^j-x)
=x=p=\RULE P^j(b).
\]
If reporting $x$ selects the other group because of priorities, $j$ is
ineligible and both rules charge zero at that report; the equality above
concerns reports at which she is eligible.

Finally, if $\RULE A(b;j)=0$, both rules induce zero expected payment.
Even if an eligible principal has a nonzero success-contingent payment
coefficient, a zero success probability makes its expected payment zero.
\end{proof}

Theorem~\ref{theorem: universally monotone allocation rule} and
Theorem~\ref{theorem: myerson} give a normalized payment rule that is
\DSIC{} for each fixed partition. By
Theorem~\ref{theorem: equivalence of myerson and monopolistic pricing},
Algorithm~\ref{algo: sampling} induces the same expected payments at
every report, and therefore the same utilities, including zero effective
allocation. This establishes the following corollary.

\begin{corollary}
\label{corollary:satosh-dsic}
    The random contract generated by \SC{} in Algorithm \ref{algo: sampling} is universally \DSIC{}.
\end{corollary}

\subsection{Approximation Analysis}

We now proceed to investigate the approximation guarantee of \SC{}. Lemma \ref{lemma: universal welfare bound} first relates the total values of eligible principals to \texttt{SumV}. Note that the logarithmic ratio provided in this proposition cannot be derived from the conventional bound provided in \cite[Lemma 3.4]{GoldbergHKSW06} on $\texttt{SumV}(b)/\mathcal{F}(b)$ directly, since a principal $j$ that is counted by the expression $\max_{k}k \cdot \pi^{k}(b)$ in Definition \ref{definition: optimal fixed price revenue} may not be eligible with the random algorithm \SC{}. 

\begin{lemma}\label{lemma: universal welfare bound}
    For any fixed coin realization, $\sum_{j \in [m]}b^{j}\cdot \delta^{j} \geq \frac{1}{2\mathcal{H}_{m}}\sum_{j \in [m]}b^{j}$, where $\mathcal{H}_{m}$ is the harmonic number with parameter $m$.
\end{lemma}

\begin{proof}
Let $W$ be the winning group and $L$ the other group under the ordered
revenues. Write $\tau^G=\mathcal F(b^G)$ for $G\in\{W,L\}$, and let
$k=\texttt{TollSharing}(b^W,\widehat{\mathcal F}(b^L))$.
Choose a rank $\hat k$ attaining the ordered maximum
$\widehat{\mathcal F}(b^W)$. 
Because $W$ wins, the revenue pair at $\hat k$ exceeds
$\widehat{\mathcal F}(b^L)$. Thus $\hat k$ is feasible and $k\geq\hat k$
by maximality. Hence,
\[
 \begin{aligned}
 \sum_{j\in[m]}b^j\delta^j
   &=\sum_{\ell=1}^{k}b^{W[\ell]}
     \geq\sum_{\ell=1}^{\hat k}b^{W[\ell]}\\
   &\geq \hat k\,b^{W[\hat k]}
     =\tau^W
     \geq \frac{\tau^W+\tau^L}{2}.
 \end{aligned}
\]
The last inequality uses $\tau^W\geq\tau^L$, which remains true when
the ordered comparison is resolved by priorities at equal revenues.
In particular, the sum of the top $\hat k$ bids is only required to be
at least their fixed-price revenue.

For either nonempty group $G$, every rank $\ell$ satisfies
$\ell b^{G[\ell]}\leq\tau^G$. Therefore,
\[
 \sum_{j\in G}b^j
 =\sum_{\ell=1}^{|G|}b^{G[\ell]}
 \leq\sum_{\ell=1}^{|G|}\frac{\tau^G}{\ell}
 =\mathcal H_{|G|}\tau^G
 \leq\mathcal H_m\tau^G.
\]
For an empty group the same inequality holds with both sides zero.
Applying it to $W$ and $L$ gives
\[
 \sum_{j\in[m]}b^j\delta^j
 \geq\frac{\tau^W+\tau^L}{2}
 \geq\frac{1}{2\mathcal H_m}
       \left(\sum_{j\in W}b^j+\sum_{j\in L}b^j\right)
 =\frac{1}{2\mathcal H_m}\sum_{j\in[m]}b^j.
\]
\end{proof}

We next bound the success probability. Since $f\big(\mathtt{TPS}(\cdot)\big)$ is nonlinear, $\mathbb{E}\big[f(\mathtt{TPS}(\tau))\big]$ need not equal $f\big(\mathtt{TPS}(\mathbb{E}[\tau])\big)$. Bounds on the expected toll alone therefore do not suffice for our analysis; we use a constant-probability lower bound on $\tau$.
Recall the benchmark $\mathcal{F}^{2}(b)$ defined in Section~\ref{sec:preliminaries}. The next lemma bounds the probability that $\tau \geq \frac{1}{4}\mathcal{F}^{2}(b)$.

\begin{lemma}\label{lemma: high probability result on toll}
    The toll extracted by \SC{} from the principals satisfies $\Pr\Big[\tau \geq \frac{1}{4}\mathcal{F}^{2}(b)\Big] \geq \frac{1}{2}$.
\end{lemma}

\begin{proof}

If $\mathcal F^2(b)=0$, the claim is immediate since $\tau\geq0$.
Otherwise choose
$k^*\in\argmax_{2\leq k\leq m}k\pi^k(b)$.
Before drawing the random partition, fix the set $H^*$ of exactly the top
$k^*$ principals in the global fixed-priority ordering. Every member of
$H^*$ has bid at least $\pi^{k^*}(b)$, even when additional principals have
that same bid. Define
$k^X=|H^*\cap X|$ and $k^Y=|H^*\cap Y|$.
Then $k^X+k^Y=k^*$ and
$\tau^X\geq\pi^{k^*}(b)k^X$,
$\tau^Y\geq\pi^{k^*}(b)k^Y$.
The toll remains $\tau=\min\{\tau^X,\tau^Y\}$ under consistent
ordering. Thus

    \begin{align*}
        \Pr\Big[\tau \geq \tfrac14\mathcal F^2(b)\Big]
        &=\Pr\Big[\min\{\tau^X,\tau^Y\}\geq\tfrac14\mathcal F^2(b)\Big]\\
        &\geq\Pr\Big[\pi^{k^*}(b)\min\{k^X,k^Y\}
                  \geq\tfrac14 k^*\pi^{k^*}(b)\Big]\\
        &=\Pr\Big[\min\{k^X,k^Y\}\geq\tfrac14 k^*\Big]\\
        &=\Pr\Big[\tfrac14 k^*\leq k^X\leq\tfrac34 k^*\Big],
    \end{align*}

The last equality follows from $k^X+k^Y=k^*$.

    Since the partition is generated by flipping a fair coin, we have $k^{X} \sim \texttt{Binomial}(k^{*}, 0.5)$. Therefore, the expected value and the variance of $k^{X}$ satisfy $\mathbb{E}[k^{X}] = \frac{1}{2}k^{*}$ and $\operatorname{Var}[k^{X}] = \frac{1}{4}k^{*}$. Using Chebyshev's inequality, we have

\begin{equation*}
        \Pr\Big[ \frac{1}{4}k^{*} \leq k^{X} \leq \frac{3}{4}k^{*} \Big] = 1 - \Pr\Big[ \big|k^{X} - \mathbb{E}[k^{X}] \big| > \frac{1}{4}k^{*}\Big] \geq 1 - \frac{\operatorname{Var}[k^{X}]}{(\frac{1}{4}k^{*})^{2}} = 1 - \frac{4}{k^{*}} \, .
    \end{equation*}

    Therefore, this proposition holds when $k^{*} \geq 8$. Note that by definition, $k^{*}$ must be at least $2$, so it remains to verify the case where $k^{*} \in [2, 8)$. For $k^{*} \in [2, 8)$, the probabilities are in Table \ref{tab: probabilities for some values}.
    
    \begin{table}
        \centering
        \caption{The probability that $k^{X} \in \Big[ \frac{1}{4}k^{*}, \frac{3}{4}k^{*} \Big]$ for the case where $k^{*} \in [2, 8)$}
        \label{tab: probabilities for some values}
        \begin{tabular}{|c|c|c|c|c|c|c|}
            \hline
            $k^{*}$ & $2$ & $3$ & $4$ & $5$ & $6$ & $7$\\
            \hline
            $\Pr\Big[ \frac{1}{4}k^{*} \leq k^{X} \leq \frac{3}{4}k^{*} \Big]$ & $\frac{1}{2}$ & $\frac{3}{4}$ & $\frac{7}{8}$ & $\frac{5}{8}$ & $\frac{25}{32}$ & $\frac{7}{8}$\\
            \hline
        \end{tabular}
    \end{table}
\end{proof}

\begin{theorem}\label{theorem: approx guarantee of satoshi}
    For $m\geq2$, \SC{} has an approximation guarantee of $\Big(O_{\alpha}(\log m), \mathcal{F}_{1/4}^{2}\Big)$.
For $m=1$, the benchmark is zero.
\end{theorem}

\begin{proof}

Let $r$ be the approximation factor of \texttt{TPS}, so
$r=O_\alpha(1)$ by Theorem~\ref{theorem: sad} and
Lemma~\ref{lemma: from jump ratio to approx factor}. For the analysis, set
$(\lambda_0,\mu_0)=(0,\emptyset)$, so the maximum over feasible tags
also represents $f(\texttt{TPS}(B))$ when $0\leq B<B_{\min}$.
When every principal bids truthfully, we have

    \begin{align*}
        \mathbb{E}\Big[ \sum_{j \in [m]}v^{j}\delta^{j} \cdot f(S) \Big] \geq& \; \frac{1}{2\mathcal{H}_{m}}\sum_{j \in [m]}v^{j} \cdot \mathbb{E}[ f(S) ] \\
        =& \; \frac{1}{2\mathcal{H}_{m}} \sum_{j \in [m]}v^{j} \cdot \mathbb{E}\Big[ \max_{k: \lambda_{k}\leq \tau} f(\mu_{k}) \Big] \\
        \geq&\; \frac{1}{2\mathcal{H}_{m}} \sum_{j \in [m]}v^{j} \cdot \mathbb{E}\Big[ \max_{k: \lambda_{k}\leq \tau} f(\mu_{k}) \Big| \tau \geq \frac{1}{4}\mathcal{F}^{2}(v) \Big] \cdot \Pr\Big[\tau \geq \frac{1}{4}\mathcal{F}^{2}(v)\Big] \\
        \geq&\; \frac{1}{4\mathcal{H}_{m}} \sum_{j \in [m]}v^{j} \cdot \mathbb{E}\Big[ \max_{k: \lambda_{k}\leq \tau} f(\mu_{k}) \Big| \tau \geq \frac{1}{4}\mathcal{F}^{2}(v) \Big] \\
        \geq& \; \frac{1}{4\mathcal{H}_{m}} \sum_{j \in [m]}v^{j} \cdot  \max_{k: \lambda_{k}\leq \frac{1}{4}\mathcal{F}^{2}(v)} f(\mu_{k})  \\
        \geq& \; \frac{1}{4r\mathcal H_m} \cdot \sum_{j \in [m]}v^{j} \cdot f\Big(S^{*}\big( \frac{1}{4}\mathcal{F}^{2}(v) \big) \Big) \, .
    \end{align*}
    
    The first transition above holds because Lemma \ref{lemma: universal welfare bound} establishes $\sum_{j \in [m]}v^{j}\delta^{j} \geq \frac{1}{2\mathcal{H}_{m}}\sum_{j \in [m]}v^{j}$ independently of the coin realization. The third one is derived from the law of total expectation. The fourth transition follows from Lemma \ref{lemma: high probability result on toll}. The fifth one holds because the term $\max_{k: \lambda_{k}\leq \tau}f(\mu_{k})$ is non-de\-creasing with $\tau$. The last transition uses the $r$-approximation guarantee of
\texttt{TPS}. Thus the expected eligible value is at least
$\mathtt{OPT}(\mathcal I,\mathcal F_{1/4}^2)/(4r\mathcal H_m)$, which gives
the claimed $O_\alpha(\log m)$ approximation for $m\geq2$; for $m=1$ the
benchmark is zero under the convention above.
\end{proof}

\clearpage
\bibliographystyle{alpha}
\bibliography{references}

\end{document}